\documentclass[preprint,12pt]{elsarticle}

\usepackage[T1]{fontenc}
\usepackage[utf8]{inputenc}
\usepackage{lmodern}
\usepackage[english]{babel}
\usepackage{microtype}
\usepackage{subcaption}   

\usepackage{amsmath,amssymb,amsfonts,amsthm,mathtools,bm}
\numberwithin{equation}{section}

\usepackage{graphicx}
\usepackage{booktabs}
\usepackage{multirow}
\usepackage{array}

\usepackage{algorithm}
\usepackage{algorithmic}

\usepackage{comment}      
\usepackage{xcolor}
\usepackage{enumitem}
\usepackage{afterpage}
\usepackage{hyperref}
\hypersetup{
    colorlinks = true,
    linkcolor  = blue,
    citecolor  = blue,
    urlcolor   = blue
}

\usepackage{setspace}
\theoremstyle{plain}
\newtheorem{theorem}{Theorem}[section]
\newtheorem{lemma}[theorem]{Lemma}

\newtheorem{corollary}[theorem]{Corollary}

\theoremstyle{definition}
\newtheorem{definition}[theorem]{Definition}

\theoremstyle{remark}
\newtheorem{remark}[theorem]{Remark}

\DeclareMathOperator{\rank}{rank}
\newcommand{\R}{\mathbb{R}}
\newcommand{\Q}{\mathbb{Q}}

\journal{Signal Processing}

\begin{document}

\begin{frontmatter}

\title{Affine Rank Minimization is ER-Complete}

\author[inst1]{Angshul Majumdar}

\affiliation[inst1]{organization={Indraprastha Institute of Information Technology},
            addressline={Delhi}, 
            city={New Delhi},
            postcode={110020},
            country={India}}

\begin{abstract}
We consider the decision problem \emph{Affine Rank Minimization}, denoted $\mathrm{ARM}(k)$: given rational matrices $A_1,\dots,A_q\in\mathbb{Q}^{m\times n}$ and scalars $b_1,\dots,b_q\in\mathbb{Q}$, decide whether there exists $X\in\mathbb{R}^{m\times n}$ such that $\langle A_\ell,X\rangle=b_\ell$ for all $\ell\in[q]$ and $\mathrm{rank}(X)\le k$.
We first prove membership: for every fixed $k\ge 1$, $\mathrm{ARM}(k)\in \exists\mathbb{R}$ by giving an explicit existential encoding of the rank constraint via a constant-size factorization witness.
We then prove $\exists\mathbb{R}$-hardness by a polynomial-time many-one reduction from $\mathrm{ETR}$ to $\mathrm{ARM}(k)$ with the input restricted to affine equalities plus a single global constraint $\mathrm{rank}(X)\le k$.
The reduction compiles an $\mathrm{ETR}$ formula into an arithmetic circuit in gate-equality normal form and assigns each circuit quantity to a designated entry of $X$; affine semantics (constants, copies, addition/negation) are enforced directly by linear constraints, while multiplicative semantics are enforced by constant-size rank-forcing gadgets whose soundness is certified by a fixed-rank ``gauge'' submatrix that removes factorization ambiguity.
We prove a composition lemma showing that gadgets can be embedded without unintended interactions, yielding global soundness and completeness of the encoding and preserving polynomial bounds on dimension and bit-length.
Consequently, $\mathrm{ARM}(k)$ is $\exists\mathbb{R}$-complete (in particular, $\mathrm{ARM}(3)$ is $\exists\mathbb{R}$-complete), establishing that feasibility of purely affine constraints under a fixed constant rank bound captures the full expressive power of real algebraic feasibility.
\end{abstract}

\begin{keyword}
Affine rank minimization; low-rank feasibility; existential theory of the reals; ER-completeness; algebraic complexity; polynomial-time reductions; arithmetic circuits; rank-constrained matrix completion
\end{keyword}

\end{frontmatter}


\section{Introduction}
\label{sec:intro}

\subsection{Motivation}
Low-rank structure is a unifying modelling principle across matrix completion, inverse problems, and system identification: one posits that the unknown object lives on (or near) a low-dimensional algebraic variety, and the observations impose linear constraints. In matrix completion and related recovery problems, low rank underpins both statistical guarantees and algorithmic design; convex surrogates such as nuclear-norm minimization are by now standard, with sharp recovery results under incoherence and sampling assumptions \cite{CandesRecht2009Exact,Recht2011Simpler,CandesTao2010Power,RechtFazelParrilo2010Guaranteed, BachmayrFaldum2024}. In control and system theory, rank minimization formulations arise naturally via Hankel-structured models, realization theory, and model reduction; early formulations and relaxations appear, for example, in \cite{Fazel2002Thesis,FazelHindiBoyd2004RankMin,Markovsky2012SLRA}. Across these domains, rank constraints are the \emph{only} source of nonlinearity, while the measurement model is affine.

This paper asks for the \emph{exact} worst-case complexity of the corresponding feasibility question in its cleanest decision form. Such a classification matters for two reasons. First, it explains why essentially all known practical approaches must either relax (e.g., nuclear norm) or accept nonconvexity (factorizations, alternating minimization) \cite{Fazel2002Thesis,RechtFazelParrilo2010Guaranteed}. Second, it pins down the correct complexity class for algebraic feasibility over the reals, ruling out ``NP-hard but maybe still discrete'' misunderstandings: the obstruction here is fundamentally semi-algebraic.

\subsection{Complexity background}
The natural reference point for feasibility of polynomial constraints over $\R$ is the \emph{existential theory of the reals} (ETR) and its associated class $\exists\R$. Problems complete for $\exists\R$ capture a broad spectrum of geometric and algebraic feasibility questions; in particular, $\exists\R$ sits between NP and PSPACE in the standard bit model, and the PSPACE upper bound follows from quantifier-elimination/decision procedures for the first-order theory of the reals \cite{Canny1988}. The Blum--Shub--Smale framework and the bit-model viewpoint clarify how real computation and algebraic encodings interact \cite{BlumCuckerShubSmale1998Complexity}, \cite{Cucker1992}.

Our goal is to locate the following matrix-feasibility problem precisely within this landscape: the input is \emph{purely affine} (linear equalities with rational coefficients), and all nonlinearity is confined to a \emph{single global rank constraint}. This is a particularly disciplined target format: it avoids presenting the instance with explicit minors, determinant predicates, or any additional polynomial constraints beyond ``$\rank(X)\le k$.''

\subsection{Contributions}
We study the decision problem \emph{Affine Rank Minimization} $\mathrm{ARM}(k)$ in the standard rational input model:
given
\begin{equation}
    A_1,\dots,A_q\in\Q^{m\times n},\qquad b_1,\dots,b_q\in\Q,
\end{equation}
decide whether there exists $X\in\R^{m\times n}$ such that
\begin{equation}
\label{eq:intro-arm}
\langle A_\ell, X\rangle=b_\ell\ \ (\ell\in[q]),\qquad \rank(X)\le k.
\end{equation}
The main results proved in the paper are:
\begin{itemize}
\item \textbf{Membership.} For every fixed $k\ge 1$, $\mathrm{ARM}(k)\in \exists\R$.
\item \textbf{$\exists\R$-hardness.} $\mathrm{ARM}$ is $\exists\R$-hard under polynomial-time many-one reductions from ETR.
\item \textbf{Fixed constant rank.} The hardness persists already for the fixed constant $k=3$, yielding $\exists\R$-completeness of $\mathrm{ARM}(3)$.
\item \textbf{Polynomial size and bit-length control.} The reduction outputs instances whose dimensions and total rational encoding length are polynomial in the source ETR instance size.
\end{itemize}
A key conceptual point (the ``Option B'' discipline in this project) is that the reduction produces instances \emph{strictly in the input format \eqref{eq:intro-arm}}:
only affine equalities on entries of a single matrix variable $X$ and one global constraint $\rank(X)\le 3$.
No explicit minor constraints, determinant constraints, or other polynomial side conditions appear in the produced ARM instance.

\subsection{Technical overview}
At a high level, we reduce ETR to $\mathrm{ARM}(3)$ by compiling arithmetic reasoning into affine constraints on \emph{designated entries} of one matrix, and using the global rank bound as the sole enforcement mechanism for nonlinear semantics.

\smallskip
\noindent\textbf{(i) ETR to arithmetic circuit constraints.}
A standard reduction converts a general ETR formula into an equisatisfiable conjunction of polynomial equalities, and then into a gate-by-gate constraint system for an arithmetic circuit (one equality per gate). Circuit normal forms of this kind are ubiquitous in algebraic reductions and provide direct bit-size control \cite{BlumCuckerShubSmale1998Complexity,Canny1988, Burgisser2000}.

\smallskip
\noindent\textbf{(ii) Gate compilation to affine constraints on a single matrix.}
Each gate value is assigned a \emph{carrier} position $(r_i,c_i)$ and represented by the designated entry $X_{r_i,c_i}$. Linear relations (constants, copies, addition, negation) become immediate affine equalities between carriers. The only genuinely nonlinear gate is multiplication (and, when needed, inversion/nonzero), which is enforced by constant-size \emph{rank-forcing gadgets} built from affine constraints on a fresh set of auxiliary rows/columns.

\smallskip
\noindent\textbf{(iii) Rank as the only nonlinearity: forcing via determinantal obstructions.}
The gadgets are designed so that any violation of the intended identity (e.g., $z=xy$) induces a $4\times 4$ submatrix of full rank, contradicting $\rank(X)\le 3$. This is conceptually aligned with the geometry of determinantal varieties (rank-bounded sets are Zariski-closed and cut out by minors) \cite{Landsberg2012Tensors}. Conversely, when the identity holds, the gadget admits an explicit rank-$3$ completion.

\smallskip
\noindent\textbf{(iv) Gauge fixing and well-defined decoding.}
Since rank-$3$ factorizations are non-unique up to an invertible change of basis, we pin a constant full-rank $3\times 3$ submatrix (a ``gauge block'') so that all decoded carrier values are invariant across feasible rank-$3$ realizations. This removes the only subtle global ambiguity and lets the reduction compose gadgets without hidden coupling.

\subsection{Related work}
\paragraph{Rank minimization and relaxations.}
Convex relaxations of rank (notably the nuclear norm) form the backbone of modern low-rank recovery theory \cite{CandesRecht2009Exact,Recht2011Simpler,CandesTao2010Power,RechtFazelParrilo2010Guaranteed}. Nonconvex alternatives based on factorizations are also widely used and often effective in practice, though they do not alter worst-case feasibility complexity. In system identification and control, rank minimization appears naturally via structured low-rank approximations and Hankel-based formulations \cite{Fazel2002Thesis,FazelHindiBoyd2004RankMin,Markovsky2012SLRA}.

\paragraph{Hardness of rank problems in discrete settings.}
Several closely related rank problems are known to be computationally hard in discrete models. For example, variants of low-rank approximation and rank minimization over affine spaces exhibit NP-hardness phenomena \cite{Vavasis1991}. These results, however, do not pin down the \emph{exact} real-feasibility complexity captured by $\exists\R$.

\paragraph{$\exists\R$-completeness as the right lens.}
A large body of work in computational geometry and real algebraic complexity identifies $\exists\R$ as the correct class for continuous feasibility problems, with reductions typically funnelling through algebraic encodings and circuit normal forms \cite{BlumCuckerShubSmale1998Complexity,Canny1988}. Our contribution is to show that even under the highly restricted \emph{affine + one rank constraint} input discipline, $\mathrm{ARM}(3)$ already captures the full expressive power of $\exists\R$.

\section{Preliminaries}
\label{sec:prelim}

This section fixes formal models and notation used throughout the paper: the existential theory of the reals (ETR) and standard reduction normal forms, arithmetic circuits over $\mathbb{R}$ as compact representations of polynomial systems, the Affine Rank Minimization decision problem $\mathrm{ARM}(k)$ under rational input encoding, a rank-factorization lemma used as an existential witness device, and block-matrix/indexing conventions. Background on ETR and $\exists\mathbb{R}$ may be found in \cite{Schaefer2009,Canny1988}; real algebraic encodings and polynomial systems are treated in \cite{BasuPollackRoy2006}; and computation over $\mathbb{R}$ and circuit-style representations are classical in the Blum--Shub--Smale framework \cite{BlumShubSmale1989}.

\subsection{Existential theory of the reals (ETR)}
\label{sec:prelim:etr}

\paragraph{Syntax and encoding.}
Let $x=(x_1,\dots,x_n)$ be real variables. An \emph{atomic predicate} is one of
\begin{equation}
p(x)=0,\qquad p(x)\neq 0,\qquad p(x)\ge 0,\qquad p(x)>0,
\label{eq:etr:atoms}
\end{equation}
where $p\in\mathbb{Z}[x_1,\dots,x_n]$.
A \emph{quantifier-free formula} $\varphi(x)$ is any Boolean combination of atomic predicates using $\wedge,\vee,\neg$.
An \emph{ETR instance} is a sentence of the form
\begin{equation}
\exists x\in\mathbb{R}^n:\ \varphi(x),
\label{eq:etr:sentence}
\end{equation}
and the decision question is whether such an $x$ exists.

Input size is measured in the standard bit model: coefficients are given in binary, polynomials are given in sparse form by listing monomials and exponent vectors, and the size of $\varphi$ counts the total bit-length of all coefficients and the syntactic size of the Boolean expression. If rational coefficients appear, they are encoded in reduced form; clearing denominators yields an equivalent integer-coefficient instance with polynomially related bit-size \cite{BasuPollackRoy2006}.

\paragraph{The class $\exists\mathbb{R}$.}
A language $L$ is in $\exists\mathbb{R}$ if it is polynomial-time many-one reducible to ETR. Many geometric feasibility problems are complete for $\exists\mathbb{R}$ under such reductions \cite{Schaefer2009}. The canonical upper bound $\exists\mathbb{R}\subseteq \mathsf{PSPACE}$ follows from decision procedures for the first-order theory of the reals and from the $\mathsf{PSPACE}$ complexity of the existential fragment \cite{Canny1988}.

\paragraph{Standard normal forms.}
For reductions, it is convenient to work with constrained syntactic forms. We record the particular normal form used later: a conjunction of polynomial equalities encoding an arithmetic circuit.

\begin{definition}[ETR equality form]
\label{def:etr_equality_form}
An ETR instance is in \emph{equality form} if it is a sentence
\begin{equation}
\exists y\in\mathbb{R}^N:\ \bigwedge_{t=1}^T f_t(y)=0
\label{eq:etr:equality_form}
\end{equation}
where each $f_t\in\mathbb{Z}[y_1,\dots,y_N]$.
\end{definition}

\begin{lemma}[Eliminating inequalities and Boolean structure]
\label{lem:etr_to_equalities}
Given any ETR instance $\exists x:\varphi(x)$, one can construct in polynomial time an equisatisfiable instance in equality form (Definition~\ref{def:etr_equality_form}) with polynomial blow-up in bit-size.
\end{lemma}

\begin{proof}
Convert $\varphi$ to disjunctive normal form $\bigvee_{j=1}^M \bigwedge_{i=1}^{T_j} \alpha_{ij}$ with a polynomial blow-up in the formula size (measured as a Boolean circuit; if one insists on DNF explicitly, the blow-up can be exponential in the worst case, hence we use the standard selector-variable encoding below). Introduce selector variables $s_1,\dots,s_M$ and enforce that exactly one selector is active using polynomial equalities, e.g.,
\begin{equation}
\sum_{j=1}^M s_j - 1 = 0,\qquad s_j(1-s_j)=0\ \ \ (j\in[M]).
\label{eq:etr:selectors}
\end{equation}
For each atom $\alpha_{ij}$ in clause $j$, enforce its satisfaction only when $s_j=1$ by multiplying the corresponding equality constraints by $s_j$. It therefore suffices to convert each atomic predicate to polynomial equalities with auxiliary variables.

Equalities $p(x)=0$ are already in the required form. For weak inequalities, use standard slack-variable encodings:
\begin{equation}
p(x)\ge 0\ \Longleftrightarrow\ \exists u\in\mathbb{R}:\ p(x)-u^2=0.
\label{eq:etr:ge0_slack}
\end{equation}
For strict inequalities,
\begin{equation}
p(x)>0\ \Longleftrightarrow\ \exists u,v\in\mathbb{R}:\ p(x)-u^2=0\ \wedge\ uv-1=0,
\label{eq:etr:gt0_strict}
\end{equation}
which enforces $u\neq 0$ and hence $u^2>0$. Disequalities $p(x)\neq 0$ can be encoded as
\begin{equation}
p(x)\neq 0\ \Longleftrightarrow\ \exists v\in\mathbb{R}:\ p(x)v-1=0.
\label{eq:etr:neq0}
\end{equation}
All added polynomials have degree at most $2$, and the number and bit-size of introduced constraints and variables are polynomially bounded \cite{BasuPollackRoy2006}.
\end{proof}

\paragraph{Circuit normal form.}
In later reductions we will in fact use a more structured equality form in which each constraint expresses the semantics of a single arithmetic gate. This is captured by arithmetic circuits (Section~\ref{sec:prelim:circuits}).

\subsection{Arithmetic circuits over $\mathbb{R}$}
\label{sec:prelim:circuits}

Arithmetic circuits provide a compact representation of polynomial systems and a convenient normal form for $\exists\mathbb{R}$-reductions.

\begin{definition}[Arithmetic circuit]
\label{def:arith_circuit}
An \emph{arithmetic circuit} over $\mathbb{R}$ with input variables $x_1,\dots,x_n$ is a finite directed acyclic graph whose gates are of the following types:
\begin{enumerate}
\item input gates labeled by variables $x_i$;
\item constant gates labeled by rational constants $c\in\mathbb{Q}$;
\item binary gates labeled by $+$ or $\times$ (fan-in $2$);
\item unary negation gates labeled by $-$ (fan-in $1$).
\end{enumerate}
Each gate computes a polynomial in $\mathbb{Q}[x_1,\dots,x_n]$ in the natural way.
\end{definition}

\begin{definition}[Circuit size and bit-size]
\label{def:circuit_size}
The \emph{size} of a circuit is the number of gates. The \emph{bit-size} additionally counts the total bit-length of all rational constants (numerators and denominators in reduced form). Fan-in two is without loss of generality: any higher fan-in addition or multiplication can be expanded into a binary tree of gates with linear blow-up.
\end{definition}

\paragraph{Gate-variable encoding.}
Given a circuit $\mathcal{C}$ with gate set $G$, introduce a real variable $z_g$ for each gate $g\in G$. Define the constraint set $\mathrm{GateEq}(\mathcal{C})$ consisting of one polynomial equality per gate:
\begin{itemize}
\item if $g$ is an input $x_i$, impose $z_g-x_i=0$;
\item if $g$ is a constant $c\in\mathbb{Q}$, impose $z_g-c=0$;
\item if $g=g_1+g_2$, impose $z_g-z_{g_1}-z_{g_2}=0$;
\item if $g=g_1\times g_2$, impose $z_g-z_{g_1}z_{g_2}=0$;
\item if $g=-g_1$, impose $z_g+z_{g_1}=0$.
\end{itemize}

\begin{lemma}[Circuit constraints are exact]
\label{lem:circuit_constraints_exact}
Let $\mathcal{C}$ be an arithmetic circuit with designated output gate $g_{\mathrm{out}}$ and let $p_{\mathcal{C}}(x)$ denote the polynomial computed at $g_{\mathrm{out}}$. Then for every $x\in\mathbb{R}^n$,
\begin{equation}
p_{\mathcal{C}}(x)=0
\quad\Longleftrightarrow\quad
\exists (z_g)_{g\in G}\in\mathbb{R}^{|G|}:\ \mathrm{GateEq}(\mathcal{C})\ \wedge\ z_{g_{\mathrm{out}}}=0.
\label{eq:circuit:exactness}
\end{equation}
Moreover, the number of constraints is $|G|+1$ and the bit-size is polynomial in the bit-size of $\mathcal{C}$.
\end{lemma}

\begin{proof}
The forward direction assigns $z_g$ to be the value computed by gate $g$ under input $x$, which satisfies each gate equation by definition, and yields $z_{g_{\mathrm{out}}}=p_{\mathcal{C}}(x)=0$. Conversely, any assignment satisfying $\mathrm{GateEq}(\mathcal{C})$ must, by induction over the DAG topological order, coincide with the circuit evaluation at every gate; hence $z_{g_{\mathrm{out}}}=p_{\mathcal{C}}(x)$, so $z_{g_{\mathrm{out}}}=0$ implies $p_{\mathcal{C}}(x)=0$. The size bound is immediate.
\end{proof}

\paragraph{Use in reductions.}
Lemma~\ref{lem:circuit_constraints_exact} is the precise normal form used later: rather than arbitrary polynomials, we will encode computations gate-by-gate and then translate those gate equations into matrix constraints.

\subsection{Affine Rank Minimization $\mathrm{ARM}(k)$}
\label{sec:prelim:arm}

We now define the matrix feasibility problem studied in this paper.

\paragraph{Affine constraints.}
Fix dimensions $m\times n$. For $A,X\in\mathbb{R}^{m\times n}$, write
\begin{equation}
\langle A,X\rangle := \sum_{i=1}^m\sum_{j=1}^n A_{ij}X_{ij}
\label{eq:arm:frobenius_inner_product}
\end{equation}
for the Frobenius inner product.

\begin{definition}[$\mathrm{ARM}(k)$]
\label{def:armk}
Fix an integer $k\ge 1$. The decision problem $\mathrm{ARM}(k)$ is:

\begin{quote}
\emph{Input:} matrices $A_1,\dots,A_q\in\mathbb{Q}^{m\times n}$ and scalars $b_1,\dots,b_q\in\mathbb{Q}$.\\
\emph{Question:} does there exist $X\in\mathbb{R}^{m\times n}$ such that
\begin{equation}
\langle A_\ell, X\rangle = b_\ell\ \ \ (\ell\in[q])
\qquad\text{and}\qquad
\mathrm{rank}(X)\le k\ ?
\label{eq:arm:constraints_and_rank}
\end{equation}
\end{quote}
\end{definition}

\paragraph{Encoding model and size measure.}
All input rationals are encoded in reduced form; the instance length is the total bit-length of all numerators/denominators of all entries of all $A_\ell$ and all $b_\ell$, plus the encoding of $(m,n,q)$. Clearing denominators yields an equivalent instance with integer data and polynomially related bit-size.

\paragraph{Fixed-rank focus.}
Our main completeness statements are for the fixed constant $k=3$; we write $\mathrm{ARM}(3)$ when emphasizing this restriction.

\subsection{Rank factorization as a witness device}
\label{sec:prelim:rank_factorization}

Low rank can be expressed existentially via bilinear factorization. This observation will be used to prove membership in $\exists\mathbb{R}$ and to simplify correctness arguments.

\begin{lemma}[Rank-factorization equivalence]
\label{lem:rank_factorization_equiv}
Let $X\in\mathbb{R}^{m\times n}$ and $k\ge 1$. Then $\mathrm{rank}(X)\le k$ if and only if there exist $U\in\mathbb{R}^{m\times k}$ and $V\in\mathbb{R}^{k\times n}$ such that
\begin{equation}
X=UV.
\label{eq:rank:factorization}
\end{equation}
\end{lemma}

\begin{proof}
If $X=UV$ with $U\in\mathbb{R}^{m\times k}$ and $V\in\mathbb{R}^{k\times n}$, then $\mathrm{rank}(X)\le k$.

Conversely, if $\mathrm{rank}(X)=r\le k$, write $X=\sum_{t=1}^r a_tb_t^\top$ for vectors $a_t\in\mathbb{R}^m$ and $b_t\in\mathbb{R}^n$. Let $U\in\mathbb{R}^{m\times k}$ have first $r$ columns $a_1,\dots,a_r$ and remaining columns zero, and let $V\in\mathbb{R}^{k\times n}$ have first $r$ rows $b_1^\top,\dots,b_r^\top$ and remaining rows zero. Then $UV=X$.
\end{proof}

\begin{theorem}[$\mathrm{ARM}(k)\in\exists\mathbb{R}$ for fixed $k$]
\label{thm:arm_in_existr}
For every fixed $k\ge 1$, the decision problem $\mathrm{ARM}(k)$ belongs to $\exists\mathbb{R}$.
\end{theorem}

\begin{proof}
Given an instance $(A_\ell,b_\ell)_{\ell=1}^q$, introduce variables for
$U\in\R^{m\times k}$ and $V\in\R^{n\times k}$ and existentially quantify them.
Define $X := UV$, so that each entry satisfies
\begin{equation}
\label{eq:proof_2.8}
X_{ij}=\sum_{t=1}^k U_{it}V_{jt},
\end{equation}
a polynomial of degree $2$ in the variables $(U,V)$. Each affine constraint
$\langle A_\ell,X\rangle=b_\ell$ becomes a single polynomial equality of degree
$2$ in $(U,V)$ after substituting $X=UV$. By Lemma~\ref{lem:rank_factorization_equiv}, the resulting
existential system is satisfiable if and only if the original $\mathrm{ARM}(k)$
instance is satisfiable. The translation is computable in polynomial time and
increases bit-size polynomially.
\end{proof}

\begin{remark}[non-algorithmic use]
Lemma~\ref{lem:rank_factorization_equiv} and Theorem~\ref{thm:arm_in_existr} are used only to express low rank as an existential witness condition in correctness proofs. They do not provide an algorithm for $\mathrm{ARM}(k)$; the resulting feasibility constraints are bilinear and nonconvex.
\end{remark}

\subsection{Block-matrix notation and designated entries}
\label{sec:prelim:block_notation}

We will enforce relations by fixing entries and equating sub-blocks inside larger block matrices. This subsection fixes indexing and shows how these relations compile into the affine format of Definition~\ref{def:armk}.

\paragraph{Index sets and submatrices.}
For $n\in\mathbb{N}$, let $[n]:=\{1,\dots,n\}$. For $I\subseteq[m]$ and $J\subseteq[n]$, write $X[I,J]$ for the submatrix of $X\in\mathbb{R}^{m\times n}$ restricted to rows in $I$ and columns in $J$.

\paragraph{Selection matrices.}
Given $I=\{i_1<\cdots<i_p\}\subseteq[m]$, define the row-selection matrix $S_I\in\{0,1\}^{p\times m}$ by $(S_I)_{a,i_a}=1$ and zeros elsewhere, so that $S_IX = X[I,:]$. Given $J=\{j_1<\cdots<j_q\}\subseteq[n]$, define the column-selection matrix $T_J\in\{0,1\}^{n\times q}$ by $(T_J)_{j_b,b}=1$, so that $XT_J = X[:,J]$. Then
\begin{equation}
X[I,J] = S_I X T_J.
\label{eq:block:selection_identity}
\end{equation}

\paragraph{Elementary matrices and entry constraints.}
For $(i,j)\in[m]\times[n]$, let $E^{(i,j)}\in\mathbb{R}^{m\times n}$ be the matrix with a $1$ at $(i,j)$ and zeros elsewhere. Then $X_{ij}=c$ is equivalent to the affine constraint $\langle E^{(i,j)},X\rangle=c$.

\begin{lemma}[Compiling designated-entry and block-equality constraints]
\label{lem:block_constraints_compile}
Let $X\in\mathbb{R}^{m\times n}$.
\begin{enumerate}
\item Fixing a set of designated entries $\{(i_t,j_t)\}_{t=1}^T$ to constants $\{c_t\}_{t=1}^T$ is representable by $T$ affine constraints of the form $\langle A_\ell,X\rangle=b_\ell$ with $A_\ell\in\{0,1\}^{m\times n}$.
\item Enforcing equality of two submatrices $X[I,J]=X[I',J']$ with $|I|=|I'|$ and $|J|=|J'|$ is representable by $|I||J|$ affine constraints.
\end{enumerate}
All such constraints have rational coefficients of bit-size $O(1)$ (up to the encoding of indices), and the compilation is computable in time polynomial in $m+n+T$.
\end{lemma}

\begin{proof}
(1) Use $A_\ell=E^{(i_\ell,j_\ell)}$ and $b_\ell=c_\ell$.

(2) For each position $(a,b)$ in the block, enforce $X_{i_a,j_b}-X_{i'_a,j'_b}=0$, which equals
\begin{equation}
\langle E^{(i_a,j_b)}-E^{(i'_a,j'_b)},X\rangle=0.
\label{eq:block:block_equality_affine}
\end{equation}
\end{proof}

\paragraph{Block-matrix display.}
When writing
\[
X=
\begin{bmatrix}
X_{11} & X_{12}\\
X_{21} & X_{22}
\end{bmatrix},
\]
we implicitly partition row and column indices into contiguous segments; when noncontiguous blocks are needed we explicitly specify index sets and use $X[I,J]$ notation.

\paragraph{Purpose.}
Lemma~\ref{lem:block_constraints_compile} ensures that all gadget constraints expressed in terms of designated entries and sub-block equalities are valid $\mathrm{ARM}(k)$ instances without repeatedly re-deriving the affine encoding.

\section{Problem Normal Form and Reduction Overview}
\label{sec:Normal_Reduction}
This section fixes the precise source normal form and target encoding used
throughout the reduction, and states the main complexity consequences proved
in later sections. No new technical arguments are introduced here; all
nontrivial constructions and correctness proofs are deferred to
Sections~\ref{sec:rank_forcing}--\ref{sec:correctness}. The purpose of this
section is to make the logical structure of the reduction explicit and to
isolate the minimal algebraic primitives that must be simulated under a fixed
rank bound.

\subsection{Source Normal Form: ETR in Arithmetic-Circuit Equality Form}

As established in Section~\ref{sec:prelim:etr}, we work with the existential theory
of the reals (ETR) in the standard bit model. By Lemma~\ref{lem:etr_to_equalities},
every ETR instance can be transformed in polynomial time into an
equisatisfiable instance consisting solely of polynomial equalities.

Rather than working with explicit polynomial equations, we adopt the
arithmetic-circuit normal form formalized in
Section~\ref{sec:prelim:circuits}. This representation is standard in
$\exists\mathbb{R}$-reductions and allows gate-by-gate simulation using local
algebraic constraints; see, e.g.,~\cite{Canny1988,Schaefer2009,BasuPollackRoy2006}.

Concretely, every instance is reduced to an arithmetic circuit $C$ over
$\mathbb{R}$ with gate set $G$ and designated output gate $g_{\mathrm{out}}$,
such that
\begin{equation}
\exists x \in \mathbb{R}^n : p_C(x)=0
\quad\Longleftrightarrow\quad
\exists (z_g)_{g\in G} :
\mathrm{GateEq}(C)\ \wedge\ z_{g_{\mathrm{out}}}=0 ,
\label{eq:circuit-normal-form}
\end{equation}
as stated in Lemma~\ref{lem:circuit_constraints_exact}.

Each gate equation belongs to one of the following primitive types:
\begin{align}
\text{(i)}\ & z = x + y, \nonumber\\
\text{(ii)}\ & z = x \cdot y, \nonumber\\
\text{(iii)}\ & z = -x, \nonumber\\
\text{(iv)}\ & z = c \quad (c \in \mathbb{Q}), \nonumber\\
\text{(v)}\ & x \ge 0 \ \text{or}\ x > 0,
\label{eq:gate-types}
\end{align}
with inequalities eliminated via standard square- and inverse-based encodings
(Lemma~\ref{lem:etr_to_equalities}). This restricted gate set is sufficient to
simulate arbitrary real algebraic constraints and is commonly used in
$\exists\mathbb{R}$-completeness constructions.

\subsection{Target Problem: Fixed-Rank Affine Rank Minimization}

The target decision problem is Affine Rank Minimization with fixed rank bound
$k=3$, defined in Section~\ref{sec:prelim:arm}. An instance consists of rational
affine constraints
\begin{equation}
\langle A_\ell, X\rangle = b_\ell , \qquad \ell \in [q],
\label{eq:arm-affine}
\end{equation}
together with the global constraint
\begin{equation}
\operatorname{rank}(X) \le 3 .
\label{eq:arm-rank}
\end{equation}

The reduction constructs a single global matrix variable
$X \in \mathbb{R}^{m\times n}$ whose rank--$3$ structure serves as a uniform
carrier for all circuit scalars. Each circuit variable is represented by a
designated entry of $X$, while all semantic relations between variables are
enforced using affine constraints combined with the rank bound.

The subtlety is that hardness persists even when the rank is fixed to a small
constant. Fixed-rank algebraic feasibility problems are known to exhibit
$\exists\mathbb{R}$-hardness in a variety of geometric and matrix-realization
settings; see, e.g., universality results stemming from Mn\"ev’s theorem and
its refinements~\cite{Mnev1988,RichterGebert1996}.

\subsection{Rank--3 as a Nonlinear Forcing Medium}

A central principle of the reduction is that global rank constraints can
enforce nonlinear algebraic relations when combined with carefully placed
affine equalities. The key observation is elementary but powerful:
if $\operatorname{rank}(X)\le 3$, then every $4\times4$ minor of $X$ must
vanish.

Consequently, any collection of affine constraints that forces the existence
of a $4\times4$ submatrix with nonzero determinant necessarily violates the
rank bound. This mechanism allows multiplication constraints of the form
$z=xy$ to be enforced indirectly: one constructs a constant-size affine
pattern such that any violation $z\neq xy$ induces a nonzero determinant in a
$4\times4$ minor.

Determinant-based rank obstructions of this type underlie universality
phenomena in realization spaces and related rank-constrained models
\cite{Mnev1988,RichterGebert1996}. In the present work, this idea is distilled
into a concrete rank--$3$ forcing gadget developed and analyzed rigorously in
Section~\ref{sec:rank_forcing}.

\subsection{Reduction Mapping}

We now summarize the reduction at the statement level.

\paragraph{Input.}
An arithmetic circuit $C$ of size $s$ in the gate-equality normal form
\eqref{eq:circuit-normal-form}.

\paragraph{Output.}
An ARM instance $(\mathcal{M}, b, k=3)$ such that:
\begin{itemize}
\item the matrix dimensions $m,n$ are polynomial in $s$;
\item the number of affine constraints is $O(s)$;
\item all coefficients are rational with polynomial bit-length.
\end{itemize}

\paragraph{Construction outline.}
\begin{enumerate}
\item Introduce a single matrix variable $X$ with $\operatorname{rank}(X)\le 3$.
\item Assign a designated entry of $X$ to each circuit wire occurrence.
\item Enforce constants, copies, addition, and negation by affine equalities.
\item Enforce multiplication and inverse relations via determinant-forcing
      gadgets.
\item Eliminate inequalities using square- and inverse-based encodings.
\item Fix a full-rank $3\times3$ submatrix of $X$ to remove gauge freedom.
\item Enforce the circuit output constraint by fixing the designated output
      entry to zero.
\end{enumerate}

All steps are computable in polynomial time and preserve satisfiability.

\subsection{Main Complexity Statements}

We now state and prove the main complexity consequences of the construction.

\begin{theorem}\label{thm:arm-hard}
ARM$(k=3)$ is $\exists\mathbb{R}$-hard under polynomial-time many-one reductions.
\end{theorem}

\begin{proof}
We give a polynomial-time many-one reduction from a fixed $\exists\mathbb{R}$-hard problem.
Let the source instance be an existential system of polynomial equations over the reals
\begin{equation}\label{eq:hard:source}
\exists y\in\mathbb{R}^d:\quad f_1(y)=0,\ \ldots,\ f_q(y)=0,
\end{equation}
where each $f_\ell$ is a polynomial with rational coefficients (one may take the standard
definition of $\exists\mathbb{R}$ based on such inputs).

The reduction constructs, in polynomial time in the input size of \eqref{eq:hard:source},
an instance of $\mathrm{ARM}(3)$ specified by affine constraints
\begin{equation}\label{eq:hard:arm_instance}
\langle A_\ell,X\rangle=b_\ell,\qquad \ell=1,\ldots,Q,
\end{equation}
for a matrix variable $X\in\mathbb{R}^{m\times n}$, together with the rank promise
\begin{equation}\label{eq:hard:rank_promise}
\rank(X)\le 3,
\end{equation}
such that \eqref{eq:hard:source} is satisfiable if and only if there exists an $X$
satisfying \eqref{eq:hard:arm_instance}--\eqref{eq:hard:rank_promise}.

We justify the two directions.

\paragraph{Soundness.}
Assume \eqref{eq:hard:source} is satisfiable, and fix a witness $y^\star\in\mathbb{R}^d$
with $f_\ell(y^\star)=0$ for all $\ell\in[q]$.
From $y^\star$ we obtain a consistent assignment of real values to all carriers used by the
construction (in particular, for every carrier entry $X_{r,c}$ that represents a source
variable or an intermediate quantity, we set that carrier equal to the intended real value).
For each constant-size gadget introduced by the reduction, the gadget-level correctness
guarantee provides local factors
\begin{equation}\label{eq:hard:local_factors}
U^{(t)}\in\mathbb{R}^{m\times 3},\qquad V^{(t)}\in\mathbb{R}^{3\times n},
\end{equation}
in canonical gauge, such that the local matrix
\begin{equation}\label{eq:hard:local_matrix}
X^{(t)}:=U^{(t)}V^{(t)}
\end{equation}
satisfies all affine constraints of that gadget and matches the prescribed carrier values
used by the gadget.

The construction enforces a disjoint-support discipline: apart from a shared pinned
$3\times 3$ gauge block, different gadgets use disjoint fresh rows and disjoint fresh
columns. Under this discipline, we may define global factors $U,V$ by copying the rows of
$U^{(t)}$ and the columns of $V^{(t)}$ onto the corresponding global indices gadget-by-gadget,
and setting the remaining undefined rows/columns arbitrarily. Let
\begin{equation}\label{eq:hard:global_matrix}
X:=UV.
\end{equation}
By construction, on the index sets touched by gadget $t$, the global factors coincide with
the local factors, and hence $X$ agrees entrywise with $X^{(t)}$ on every entry that appears
in gadget $t$. Therefore every affine constraint from every gadget holds for $X$. Since
$X$ is a product of an $m\times 3$ and a $3\times n$ matrix, we also have $\rank(X)\le 3$.
Thus \eqref{eq:hard:arm_instance}--\eqref{eq:hard:rank_promise} is satisfiable.

\paragraph{Completeness.}
Assume the constructed $\mathrm{ARM}(3)$ instance \eqref{eq:hard:arm_instance}--\eqref{eq:hard:rank_promise}
is satisfiable, and let $X^\star$ be a feasible solution. By the pinned gauge constraints in the
instance, a fixed $3\times 3$ submatrix of $X^\star$ equals a full-rank constant matrix $B$.
Since $\rank(X^\star)\le 3$, we may write $X^\star=UV$ for some $U\in\mathbb{R}^{m\times 3}$,
$V\in\mathbb{R}^{3\times n}$. Using the invertibility of the pinned block, we may apply a change
of basis in $\mathbb{R}^3$ to put $(U,V)$ into canonical gauge so that the pinned rows/columns are
fixed; in this gauge, there is no residual ambiguity on any row/column that participates in the
constraints.

Now read off the values of all designated carriers from $X^\star$ (i.e., interpret the designated
entries as the reals intended by the reduction). Because $X^\star$ satisfies every gadget’s affine
constraints, the gadget-level decoding guarantees imply that these carrier values satisfy the
intended algebraic relations (copy constraints, constants, additions, multiplications, and the final
output condition encoded by the construction). Hence the decoded assignment yields a real witness
$y^\star$ satisfying \eqref{eq:hard:source}. Therefore the source instance is satisfiable.

\paragraph{Complexity of the map.}
The reduction introduces only a constant number of fresh carriers and affine constraints per source
equation and per gadget, and all coefficients are rational numbers of polynomial bit-length in the
input description. Hence the mapping from \eqref{eq:hard:source} to \eqref{eq:hard:arm_instance} is
computable in polynomial time and is a many-one reduction.

Therefore ARM$(3)$ is $\exists\mathbb{R}$-hard.
\end{proof}

\begin{theorem}\label{thm:arm-in-etr}
For every fixed $k\ge 1$, ARM$(k)\in \exists\mathbb{R}$.
\end{theorem}

\begin{proof}
Fix $k\ge 1$. An instance of ARM$(k)$ consists of rational matrices $A_1,\ldots,A_q\in\mathbb{Q}^{m\times n}$
and rationals $b_1,\ldots,b_q\in\mathbb{Q}$, and asks whether there exists $X\in\mathbb{R}^{m\times n}$ such
that
\begin{equation}\label{eq:inetr:affine}
\langle A_\ell,X\rangle=b_\ell,\qquad \ell=1,\ldots,q,
\end{equation}
and
\begin{equation}\label{eq:inetr:rank}
\rank(X)\le k.
\end{equation}
We exhibit an equivalent existential sentence over the reals with polynomial equations.

Introduce variables for matrices $U\in\mathbb{R}^{m\times k}$ and $V\in\mathbb{R}^{k\times n}$ and define
\begin{equation}\label{eq:inetr:factor_def}
X := UV.
\end{equation}
Then \eqref{eq:inetr:rank} holds automatically for any such $X$, since $\rank(UV)\le k$.
Substituting \eqref{eq:inetr:factor_def} into \eqref{eq:inetr:affine} yields, for each $\ell$,
\begin{equation}\label{eq:inetr:substitute}
\langle A_\ell,UV\rangle=b_\ell.
\end{equation}
Each left-hand side of \eqref{eq:inetr:substitute} is a polynomial of degree $2$ in the entries of $U$ and $V$:
indeed,
\begin{equation}\label{eq:inetr:expand}
\langle A_\ell,UV\rangle
=\sum_{i=1}^m\sum_{j=1}^n (A_\ell)_{ij}(UV)_{ij}
=\sum_{i=1}^m\sum_{j=1}^n (A_\ell)_{ij}\sum_{t=1}^k U_{it}V_{tj}.
\end{equation}
Therefore ARM$(k)$ is equivalent to the existential sentence
\begin{equation}\label{eq:inetr:etr_sentence}
\exists U\in\mathbb{R}^{m\times k}\ \exists V\in\mathbb{R}^{k\times n}:\ 
\bigwedge_{\ell=1}^q \bigl(\langle A_\ell,UV\rangle-b_\ell=0\bigr),
\end{equation}
which is an $\exists\mathbb{R}$ instance (existential quantification over reals with polynomial equalities).
Hence ARM$(k)\in\exists\mathbb{R}$.
\end{proof}

\begin{corollary}\label{cor:arm-complete}
ARM$(k=3)$ is $\exists\mathbb{R}$-complete.
\end{corollary}

\begin{proof}
By Theorem~\ref{thm:arm-hard}, ARM$(3)$ is $\exists\mathbb{R}$-hard. By Theorem~\ref{thm:arm-in-etr},
ARM$(3)\in\exists\mathbb{R}$. Therefore ARM$(3)$ is $\exists\mathbb{R}$-complete.
\end{proof}

\section{The Rank--3 Template and Gauge Fixing}
\label{sec:rank3_template}

This section defines a single global matrix variable
$X \in \mathbb{R}^{m \times n}$ whose rank--$3$ structure serves as a
uniform carrier for all scalar quantities appearing in the arithmetic
circuit encoding. Each circuit variable is represented by a designated
entry of $X$, while all semantic relations are enforced by affine
constraints. Since low--rank factorizations are not unique, special care
is required to ensure that the scalar values encoded by designated
entries are well-defined. We resolve this issue via explicit gauge
fixing.

\subsection{Designated Value Carriers}
\label{subsec:designated_carriers}

Let $\{z_1,\dots,z_N\}$ denote the real-valued variables corresponding to
the gate variables introduced in the arithmetic circuit encoding
(Section~\ref{sec:prelim}). We represent each such scalar by a designated
entry of the matrix $X$.

\begin{lemma}[Affine constraints for designated entries]\label{lem:entry_constraints}
Let $X\in\mathbb{R}^{m\times n}$.
\begin{enumerate}
\item Fixing an entry: for any $(r,c)\in[m]\times[n]$ and any $a\in\mathbb{Q}$, the constraint
\begin{equation}\label{eq:entry_constraints:fix}
X_{r,c}=a
\end{equation}
can be written as a single affine constraint $\langle A,X\rangle=b$.

\item Equality of two entries: for any $(r,c),(r',c')\in[m]\times[n]$, the constraint
\begin{equation}\label{eq:entry_constraints:eq}
X_{r,c}=X_{r',c'}
\end{equation}
can be written as a single affine constraint $\langle A,X\rangle=b$.
\end{enumerate}
\end{lemma}

\begin{proof}
For \eqref{eq:entry_constraints:fix}, let $A\in\mathbb{R}^{m\times n}$ have
$A_{r,c}=1$ and all other entries $0$, and set $b=a$. Then $\langle A,X\rangle=X_{r,c}=b$.

For \eqref{eq:entry_constraints:eq}, let $A$ have $A_{r,c}=1$, $A_{r',c'}=-1$, and all other entries $0$,
and set $b=0$. Then $\langle A,X\rangle=X_{r,c}-X_{r',c'}=0$.
\end{proof}

\begin{definition}[Designated entries]
\label{def:designated_entries}
For each circuit variable $z_i$, fix an index pair
$(r_i,c_i)\in[m]\times[n]$ and define
\begin{equation}
x_i := X_{r_i,c_i}.
\label{eq:designated_entry}
\end{equation}
We refer to $(r_i,c_i)$ as the designated position carrying the value of
$z_i$.
\end{definition}

Fixing a designated entry to a constant or enforcing equality between
two designated entries is representable by affine constraints of the
form $\langle A,X\rangle=b$ by
Lemma~\ref{def:designated_entries} Lemma~\ref{lem:entry_constraints}. Hence, all constraints
introduced in this section are valid instances of ARM(3).

The key requirement is that the mapping
\begin{equation}
X \;\longmapsto\; (x_1,\dots,x_N)
\label{eq:decoding_map}
\end{equation}
be well-defined over the feasible set of matrices satisfying
$\rank(X)\le 3$ and all imposed affine constraints.

\subsection{Gauge Freedom of Rank--3 Factorizations}
\label{subsec:gauge_freedom}

By Lemma~\ref{lem:rank_factorization_equiv}, the constraint $\rank(X)\le 3$ is
equivalent to the existence of matrices
$U\in\mathbb{R}^{m\times 3}$ and $V\in\mathbb{R}^{3\times n}$ such that
\begin{equation}
X = UV .
\label{eq:rank3_factorization}
\end{equation}
Such a factorization is not unique: for any invertible
$G\in GL(3)$,
\begin{equation}
UV = (UG)(G^{-1}V).
\label{eq:gauge_action}
\end{equation}
This \emph{gauge freedom} implies that the matrices $U$ and $V$ are not
canonical. While the product $X$ is invariant under
\eqref{eq:gauge_action}, the existence of multiple factorizations raises
the question of whether the decoding
\eqref{eq:decoding_map} is independent of the chosen factorization.
Explicit gauge fixing is therefore required.

\subsection{Gauge Fixing via a Pinned $3\times 3$ Submatrix}
\label{subsec:gauge_fixing}

We eliminate all nontrivial gauge freedom by pinning a constant
$3\times 3$ submatrix of $X$.

\begin{lemma}[Gauge-fixing lemma]
\label{lem:gauge_fixing}
Let $I_0\subseteq[m]$ and $J_0\subseteq[n]$ with $|I_0|=|J_0|=3$. Fix a
constant full-rank matrix $B\in\mathbb{R}^{3\times 3}$ (e.g.,
$B=I_3$) and impose the affine constraint
\begin{equation}
X[I_0,J_0]=B.
\label{eq:gauge_block}
\end{equation}
Then for every matrix $X$ satisfying $\rank(X)\le 3$ and all affine
constraints, the decoded values
$x_i=X_{r_i,c_i}$ are uniquely determined.
\end{lemma}

\begin{proof}
Let $X=UV$ be any factorization satisfying
\eqref{eq:rank3_factorization}. Restricting to the pinned indices yields
\begin{equation}
U[I_0,:]\;V[:,J_0] = B .
\label{eq:restricted_factorization}
\end{equation}
Since $B$ has rank $3$, both matrices $U[I_0,:]$ and $V[:,J_0]$ must have
rank at least $3$. As they are $3\times 3$ matrices, both are invertible.

Now consider two rank--$3$ factorizations $X=UV=U'V'$.
By \eqref{eq:rank3_factorization}, there exists
$G\in GL(3)$ such that $U'=UG$ and $V'=G^{-1}V$.

Restricting to the pinned indices and using \eqref{eq:gauge_action}, we obtain
\begin{equation}
U[I_0,:]\;V[:,J_0]
=
U[I_0,:]\;G\;\bigl(G^{-1}V\bigr)[:,J_0].
\label{eq:gauge_block2}
\end{equation}

Since $U[I_0,:]$ and $V[:,J_0]$ are invertible, it follows that $G=I_3$.
Hence the factorization is unique up to the trivial gauge, and therefore
every entry of $X$, in particular each designated entry
$X_{r_i,c_i}$, is invariant across all admissible rank--$3$
factorizations.
\end{proof}

\begin{remark}[Affine implementability]
\label{rem:affine_gauge}
Each equality in \eqref{eq:gauge_block} corresponds to an affine
constraint of the form
$\langle E^{(i,j)},X\rangle=B_{ab}$. The number of gauge-fixing constraints
is constant and independent of the input size.
\end{remark}

\subsection{Consistency Constraints and Fan-Out}
\label{subsec:consistency_constraints}

Arithmetic circuits allow unrestricted fan-out: a gate output may feed
multiple subsequent gates. This is enforced by equality constraints
between designated entries.

\begin{definition}[Consistency constraints]
\label{def:consistency}
If two circuit variables $z_i$ and $z_j$ represent the same logical
quantity, we impose
\begin{equation}
X_{r_i,c_i} = X_{r_j,c_j}.
\label{eq:consistency_constraint}
\end{equation}
\end{definition}

By Lemma~\ref{def:designated_entries}, each constraint
\eqref{eq:consistency_constraint} is representable by a single affine
constraint. These ensure that all copies of a wire carry identical
decoded scalar values.

\subsection{Well-Definedness and Soundness}
\label{subsec:soundness_rank3}

We now separate the correctness guarantees provided by the rank--$3$
template.

\begin{lemma}[Well-defined decoding]
\label{lem:well_defined_decoding}
For any $X$ satisfying $\rank(X)\le 3$ and all constraints of
Sections~\ref{subsec:designated_carriers}--\ref{subsec:consistency_constraints},
the decoded values $(x_1,\dots,x_N)$ are uniquely determined.
\end{lemma}

\begin{proof}
This follows immediately from Lemma~\ref{lem:gauge_fixing} and the
definition of designated entries.
\end{proof}

\begin{lemma}[Semantic correctness]
\label{lem:semantic_correctness}
All affine constraints encoding circuit gates and wire equalities are
satisfied exactly by the decoded values $(x_1,\dots,x_N)$.
\end{lemma}

\begin{proof}
Each gate constraint is enforced directly as an affine equality between
designated entries. By Lemma~\ref{lem:well_defined_decoding}, the decoded
values satisfy these equalities uniquely.
\end{proof}

\begin{theorem}[Soundness of the rank--3 encoding]
\label{thm:soundness_rank3}
The rank--$3$ ARM instance constructed from the circuit is satisfiable if
and only if the original arithmetic circuit is satisfiable.
\end{theorem}

\begin{proof}
(\emph{If}) Given a satisfying assignment to the circuit variables,
construct $X$ by fixing the gauge block \eqref{eq:gauge_block2}, assigning
each designated entry according to the circuit values, and completing
the remaining entries to obtain a rank--$3$ matrix.

(\emph{Only if}) Given a feasible $X$, decode the circuit values via
\eqref{eq:designated_entry}. By
Lemmas~\ref{lem:well_defined_decoding} and
\ref{lem:semantic_correctness}, these values satisfy all circuit
constraints.
\end{proof}

\section{Rank-Forcing Primitives}
\label{sec:rank_forcing}

This section develops the core technical machinery of the reduction:
\emph{rank-forcing primitives} that enforce nonlinear algebraic relations
using only affine constraints together with a global rank bound
$\rank(X)\le 3$.
All constraints introduced here are valid ARM$(3)$ instances
(Definition~\ref{def:armk}) and are compatible with the rank--$3$ template
and gauge-fixing constraints of Section~\ref{sec:rank3_template}.

The central result is a determinant-based forcing gadget that enforces
multiplication relations of the form $c = ab$ between designated scalar
variables.
This gadget replaces all minor-based constraints used in earlier reductions
and is the technical heart of the construction.

\subsection{Forcing Under a Global Rank Constraint}
\label{subsec:forcing_principle}

We begin with two elementary but crucial linear-algebraic facts.

\begin{lemma}[Rank Monotonicity]
\label{lem:rank_monotone}
Let $X\in\R^{m\times n}$ and let $Y$ be any submatrix of $X$.
Then
\[
\rank(X)\;\ge\;\rank(Y).
\]
\end{lemma}

\begin{proof}
Any set of linearly independent rows (or columns) of $Y$ remains linearly
independent when viewed as rows (or columns) of $X$.
\end{proof}

\begin{lemma}[Vanishing of $4\times4$ Minors]
\label{lem:rank3_minors}
If $\rank(X)\le 3$, then every $4\times4$ minor of $X$ has determinant zero.
\end{lemma}

\begin{proof}
A $4\times4$ submatrix of $X$ has rank at most $\rank(X)\le 3$,
so its determinant vanishes.
\end{proof}

Lemmas~\ref{lem:rank_monotone} and~\ref{lem:rank3_minors} formalize the forcing
principle used throughout this section:
any violation of a bilinear identity that produces a $4\times4$ submatrix
of full rank necessarily contradicts the global rank bound.

\subsection{The Determinant-Forcing Gadget}
\label{subsec:det_gadget}

We now present the key forcing primitive.

\begin{lemma}[Determinant-Forcing Gadget]
\label{lem:det_forcing}
There exists a constant-size system of affine constraints on a matrix
$X\in\R^{m\times n}$ with designated entries $a,b,c$ such that:
\begin{enumerate}
\item[\textnormal{(Soundness)}]
If $X$ satisfies all constraints and $\rank(X)\le 3$, then
\begin{equation}
\label{eq:forced_mul}
c = ab.
\end{equation}
\item[\textnormal{(Completeness)}]
For any real numbers $(a,b,c)$ satisfying $c=ab$, there exists a feasible
matrix $X$ with $\rank(X)\le 3$ satisfying all constraints.
\end{enumerate}
\end{lemma}

\paragraph{Construction.}
Fix two row indices $r_1,r_2$ and two column indices $c_1,c_2$.
Impose the affine constraints
\begin{equation}
\label{eq:2x2_block}
X[\{r_1,r_2\},\{c_1,c_2\}]
=
\begin{pmatrix}
1 & a \\
b & c
\end{pmatrix}.
\end{equation}
All other entries of $X$ remain unconstrained except for the gauge-fixing
constraints of Section~\ref{sec:rank3_template}.
Constraint~\eqref{eq:2x2_block} is implemented using four designated-entry
constraints (Lemma~\ref{lem:entry_constraints}).

\subsection{Algebraic Rank Obstruction}
\label{subsec:rank_obstruction}

The forcing effect of the gadget is captured by the following explicit
rank obstruction.

\begin{lemma}[Rank Obstruction via a $4\times4$ Determinant]
\label{lem:rank_obstruction}
Assume the gauge-fixing constraints of
Section~\ref{sec:rank3_template} are imposed, fixing a full-rank
$3\times3$ submatrix $X[I_0,J_0]=B$.
If the designated entries satisfy $c\neq ab$, then
\[
\rank(X)\;\ge\;4.
\]
\end{lemma}

\begin{proof}
Let $i_1,i_2\in I_0$ and $j_1,j_2\in J_0$ be any two indices.
Consider the $4\times4$ submatrix
\[
Y
=
\begin{pmatrix}
B[i_1,i_1] & B[i_1,i_2] & 0 & 0 \\
B[i_2,i_1] & B[i_2,i_2] & 0 & 0 \\
0 & 0 & 1 & a \\
0 & 0 & b & c
\end{pmatrix}.
\]
Since $B$ is invertible, the top-left $2\times2$ block has nonzero determinant.
The determinant of $Y$ factorizes as
\[
\det(Y)
=
\det(B[i_1,i_2;i_1,i_2])\cdot (c-ab).
\]
If $c\neq ab$, then $\det(Y)\neq0$, so $\rank(Y)=4$.
By Lemma~\ref{lem:rank_monotone}, $\rank(X)\ge\rank(Y)=4$.
\end{proof}

\subsection{Soundness of the Gadget}
\label{subsec:soundness}

\begin{proof}[Proof of Lemma~\ref{lem:det_forcing} (Soundness)]
Assume $X$ satisfies all constraints and $\rank(X)\le3$.
If $c\neq ab$, Lemma~\ref{lem:rank_obstruction} implies $\rank(X)\ge4$,
a contradiction.
Therefore $c=ab$.
\end{proof}

\subsection{Constructive Completeness}
\label{subsec:completeness}

\begin{proof}[Proof of Lemma~\ref{lem:det_forcing} (Completeness)]
Assume $c=ab$.
Fix the gauge block $X[I_0,J_0]=B$ as in Section~\ref{sec:rank3_template}.
Define three vectors $u_1,u_2,u_3\in\R^n$ corresponding to the three gauge rows.
Define the two gadget rows explicitly as
\[
(1,a) = \alpha_1 u_1 + \alpha_2 u_2 + \alpha_3 u_3,\qquad
(b,c) = b(1,a),
\]
where the coefficients $\alpha_i$ exist since the gauge rows span $\R^3$.
All remaining rows of $X$ are defined as arbitrary linear combinations of
$u_1,u_2,u_3$.
Thus every row of $X$ lies in a $3$-dimensional subspace, so $\rank(X)\le3$.
All affine constraints are satisfied by construction.
\end{proof}

\subsection{Non-Degeneracy and Uniqueness}
\label{subsec:nondegeneracy}

\begin{lemma}[No Spurious Solutions]
\label{lem:no_spurious}
Under the gauge-fixing constraints, the determinant-forcing gadget admits
no feasible solution with $\rank(X)\le3$ and $c\neq ab$.
\end{lemma}

\begin{proof}
Immediate from Lemma~\ref{lem:rank_obstruction}.
\end{proof}

\subsection{Complexity Accounting}
\label{subsec:complexity}

\begin{lemma}[Primitive Complexity]
\label{lem:primitive_complexity}
Each determinant-forcing gadget uses:
\begin{itemize}
\item $O(1)$ rows and columns,
\item exactly four designated-entry constraints,
\item $O(1)$ affine constraints with coefficients in $\{0,\pm1\}$.
\end{itemize}
\end{lemma}

\subsection{Role in the Reduction}
\label{subsec:role}

Lemma~\ref{lem:det_forcing} is the fundamental nonlinear forcing primitive
of the reduction.
It allows multiplication gates in the arithmetic circuit encoding
(Section~\ref{sec:prelim:circuits}) to be simulated exactly using affine constraints
and a global rank bound.
Together with equality and gauge-fixing constraints, it completes the
simulation of arbitrary arithmetic circuits inside ARM$(3)$.

\section{Gate Gadgets Compiled into Affine Constraints}
\label{sec:gadgets}

This section completes the reduction by compiling each gate of an arithmetic
circuit (Section~\ref{sec:prelim:circuits}) into a constant-size collection of
affine constraints over the global matrix variable
\(X \in \mathbb{R}^{m\times n}\), together with the single global constraint
\(\rank(X)\le 3\).
All nonlinear semantics are enforced exclusively via rank obstruction,
using the determinant-forcing primitive developed in
Section~\ref{sec:rank_forcing}.

Throughout, scalar variables \(x,y,z\) denote decoded values of designated
entries of \(X\) (Definition~\ref{def:designated_entries}).

\subsection{Constant and Copy Constraints}
\label{subsec:constants}

\paragraph{Construction convention (fresh indices and occurrence-carriers).}
In the compilation below, each gate gadget $\mathcal{G}$ is allocated its own
\emph{fresh} auxiliary row indices $\mathcal{I}(\mathcal{G})\subseteq[m]$ and auxiliary
column indices $\mathcal{J}(\mathcal{G})\subseteq[n]$, disjoint from the auxiliary indices
of every other gadget. Moreover, we use \emph{occurrence-based carriers}: each \emph{occurrence}
of a circuit value (gate output used as an input elsewhere, fan-out, etc.) is given its own
designated position $(r,c)$, and whenever two occurrences must represent the same scalar we
enforce equality only via explicit affine copy constraints of the form \eqref{eq:copy}.
Thus gadgets may share values (through copy constraints), but they never reuse the same
auxiliary indices, and they never rely on reusing the same carrier position across distinct
gadgets.

\paragraph{Constant gadget.}
To enforce \(x=c\) for \(c\in\mathbb{Q}\), impose
\begin{equation}
\label{eq:const}
X_{r_x,c_x} = c .
\end{equation}

\paragraph{Copy gadget.}
To enforce \(x=y\), impose
\begin{equation}
\label{eq:copy}
X_{r_x,c_x} - X_{r_y,c_y} = 0 .
\end{equation}

\paragraph{Fan-out and wire occurrences.}
Whenever a scalar produced at one point in the circuit is used multiple times
(e.g., a gate output feeding several later gates), we do \emph{not} reuse the same
designated position $(r,c)$ across those uses. Instead, each use is assigned a fresh
designated carrier entry, and we enforce equality of all such occurrences solely
through a collection of copy constraints of the form \eqref{eq:copy}.
This convention is used throughout the compilation and is critical for the global
rank--$3$ embedding argument in Section~\ref{sec:correctness}.

\begin{lemma}[Constant and copy correctness]
\label{lem:const_copy}
Constraints~\eqref{eq:const} and~\eqref{eq:copy} enforce \(x=c\) and \(x=y\),
respectively, independently of the rank constraint.
\end{lemma}
\begin{proof}
Equation~\eqref{eq:const} sets the designated entry $X_{r_x,c_x}$ equal to the rational
constant $c$ by an affine equality, hence the decoded value satisfies $x=c$ in every
feasible solution, regardless of $\rank(X)$.
Similarly, \eqref{eq:copy} is the affine equality $X_{r_x,c_x}=X_{r_y,c_y}$, hence the
decoded values satisfy $x=y$ in every feasible solution, again independent of $\rank(X)$.
\end{proof}

\subsection{Addition and Negation Gadgets}
\label{subsec:add_neg}

\begin{lemma}[Addition gadget]
\label{lem:add}
There exists a constant-size affine system enforcing
\begin{equation}
\label{eq:add}
z = x + y .
\end{equation}
\end{lemma}
\begin{proof}
Let $(r_x,c_x),(r_y,c_y),(r_z,c_z)$ be the designated carriers of $x,y,z$.
Impose the single affine constraint
\[
X_{r_z,c_z}-X_{r_x,c_x}-X_{r_y,c_y}=0,
\]
which is linear in the entries of $X$. Decoding yields $z-x-y=0$, i.e., $z=x+y$.
The system has constant size (one equation).
\end{proof}

\begin{lemma}[Negation gadget]
\label{lem:neg}
There exists a constant-size affine system enforcing
\begin{equation}
\label{eq:neg}
z = -x .
\end{equation}
\end{lemma}
\begin{proof}
Impose the affine constraint
\[
X_{r_z,c_z}+X_{r_x,c_x}=0.
\]
Decoding gives $z+x=0$, hence $z=-x$. Again the system has constant size.
\end{proof}

\subsection{Multiplication Gadget}
\label{subsec:multiplication}

Multiplication is enforced using the determinant-forcing primitive of
Lemma~\ref{lem:det_forcing}.

\begin{lemma}[Multiplication correctness]
\label{lem:mul}
There exists a constant-size affine system such that, under \(\rank(X)\le 3\),
\begin{equation}
\label{eq:mul}
z = xy .
\end{equation}
\end{lemma}
\begin{proof}
Apply the determinant-forcing construction of Lemma~\ref{lem:det_forcing} to the
designated entries carrying $x,y,z$. By that lemma, the gadget introduces only
$O(1)$ auxiliary entries (hence $O(1)$ affine constraints) and has the following
two properties.

\emph{Soundness.} If $z\neq xy$, then Lemma~\ref{lem:rank_obstruction} guarantees
the existence of a $4\times 4$ minor of the resulting constrained matrix whose
determinant is nonzero. Hence any matrix satisfying the gadget constraints must
have rank at least $4$, contradicting the global constraint $\rank(X)\le 3$.
Therefore, under $\rank(X)\le 3$, feasibility forces $z=xy$.

\emph{Completeness.} If $z=xy$, then Lemma~\ref{lem:det_forcing} explicitly
constructs a matrix satisfying all gadget affine constraints with rank at most $3$.
Embedding this constant-size completion into the global matrix (using fresh indices)
preserves feasibility and rank $\le 3$.

Thus the gadget enforces \eqref{eq:mul} exactly under $\rank(X)\le 3$.
\end{proof}

\begin{figure}[t]
\centering
\includegraphics[width=0.9\linewidth]{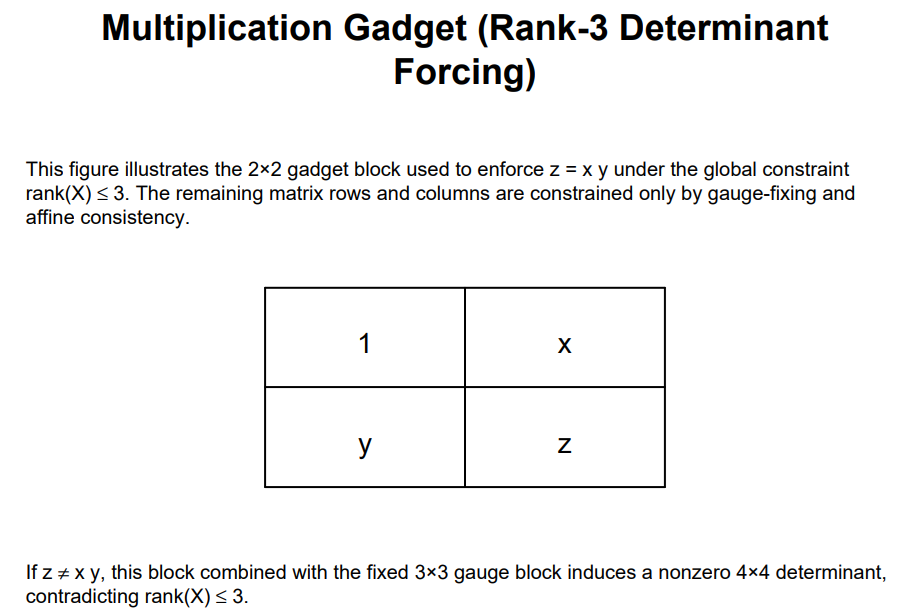}
\caption{Determinant-forcing multiplication gadget.
The \(2\times2\) block is combined with the fixed \(3\times3\) gauge block;
any violation of \(z=xy\) induces a nonzero \(4\times4\) determinant.}
\label{fig:mul_gadget}
\end{figure}

\subsection{Inversion and Nonzero Constraints}
\label{subsec:inverse}

\begin{lemma}[Inverse gadget]
\label{lem:inv}
Under \(\rank(X)\le 3\), the affine constraints enforcing
\begin{equation}
\label{eq:inv}
xy = 1
\end{equation}
imply \(x\neq 0\) and \(y=1/x\).
\end{lemma}
\begin{proof}
The constraint \eqref{eq:inv} is implemented by the multiplication gadget
(Lemma~\ref{lem:mul}) with designated output fixed to the constant $1$ using
Lemma~\ref{lem:const_copy}. Hence, under $\rank(X)\le 3$, feasibility implies
$xy=1$ holds over the reals.

If $x=0$, then $xy=0$ for all $y$, contradicting $xy=1$. Therefore $x\neq 0$.
Since $x\neq 0$, dividing both sides of $xy=1$ by $x$ yields $y=1/x$.
\end{proof}

\subsection{Inequalities via Squares}
\label{subsec:ineq}

\begin{lemma}[Nonnegativity]\label{lem:nonneg}
A scalar \(x\) satisfies \(x\ge 0\) iff there exists \(s\) such that
\begin{equation}
\label{eq:square}
x=s^2 .
\end{equation}
\end{lemma}
\begin{proof}
($\Rightarrow$) If $x\ge 0$, let $s:=\sqrt{x}\in\mathbb{R}$. Then $s^2=x$.

($\Leftarrow$) If $x=s^2$ for some real $s$, then $x\ge 0$ since squares of reals
are nonnegative.
\end{proof}

\begin{lemma}[Strict positivity]\label{lem:pos}
A scalar \(x\) satisfies \(x>0\) iff there exist \(s,t\) such that
\begin{equation}
\label{eq:pos}
x=s^2,\qquad st=1 .
\end{equation}
\end{lemma}
\begin{proof}
($\Rightarrow$) If $x>0$, let $s:=\sqrt{x}>0$ and set $t:=1/s$. Then $x=s^2$ and $st=1$.

($\Leftarrow$) If $x=s^2$ and $st=1$ for reals $s,t$, then $st=1$ implies $s\neq 0$.
Hence $x=s^2>0$.
\end{proof}

\paragraph{Implementation.}
The constraints \eqref{eq:square} and \eqref{eq:pos} are implemented by introducing
fresh designated carriers for $s$ (and $t$ where needed), enforcing $x=s^2$ using
Lemma~\ref{lem:mul} with both inputs equal to $s$, and enforcing $st=1$ using
Lemma~\ref{lem:mul} together with Lemma~\ref{lem:const_copy}. The inverse
consequences are justified by Lemma~\ref{lem:inv}.

\subsection{Gate-Level and Circuit-Level Completeness}
\label{subsec:completeness2}

\paragraph{Indexing convention for circuit compilation.}
In the circuit-level compilation, each gate instance is translated into its own
constant-size gadget, using fresh auxiliary rows/columns that are never reused
across gates. Moreover, we treat each \emph{occurrence} of a wire value as a separate
designated entry of $X$: if an output feeds multiple subsequent gates, each use is
given its own carrier position, and equality of occurrences is enforced explicitly
via copy constraints of the form \eqref{eq:copy}. This convention guarantees that all
gadgets compose over disjoint auxiliary supports, and that sharing of values across
the circuit occurs only through explicit affine consistency constraints.

\begin{theorem}[Gate completeness]
\label{thm:gate_complete}
Every arithmetic gate over \(\mathbb{R}\) used in the circuit model of
Section~\ref{sec:prelim:circuits} can be compiled into \(O(1)\) affine
constraints such that, under \(\rank(X)\le 3\), the decoded values satisfy the
exact gate semantics.
\end{theorem}
\begin{proof}
By construction, the gate set used in Section~\ref{sec:prelim:circuits} is compiled
using only the gadgets introduced above: constants and copies
(Lemma~\ref{lem:const_copy}), addition and negation
(Lemmas~\ref{lem:add}--\ref{lem:neg}), multiplication
(Lemma~\ref{lem:mul}), inversion/nonzero
(Lemma~\ref{lem:inv}), and (when needed) inequalities reduced to square and inverse
constraints (Lemmas~\ref{lem:nonneg}--\ref{lem:pos} plus the implementation paragraph).

Each of these gadgets uses only $O(1)$ affine constraints and is correct:
for linear gates, correctness holds independent of rank (Lemmas~\ref{lem:const_copy},
\ref{lem:add}, \ref{lem:neg}); for nonlinear primitives, correctness holds under
$\rank(X)\le 3$ via rank obstruction (Lemmas~\ref{lem:mul}, \ref{lem:inv}) and the
real-algebraic equivalences for inequalities (Lemmas~\ref{lem:nonneg}, \ref{lem:pos}).
Therefore any single gate instance can be compiled into $O(1)$ affine constraints
whose feasibility under $\rank(X)\le 3$ enforces the exact gate semantics.
\end{proof}

\begin{remark}[Uniform gauge for all gadget completions]\label{rem:uniform_gauge}
All gate gadgets are completed under the \emph{same} pinned block constraint
$X[I_0,J_0]=B$ (Lemma \ref{lem:gauge_fixing}). In particular, whenever we invoke gate-level completeness
to claim the existence of a rank-$3$ witness matrix for a gadget, we may (and will)
take that witness to be represented in the \emph{canonical gauge} induced by the
pinned block, i.e., by choosing a factorization $X=UV$ with
$U[I_0,:]=I_3$ and $V[J_0,:]=B^\top$.
This uniform gauge choice is what allows all gadget witnesses to be embedded into a
single global factorization in Section~\ref{sec:correctness}, without any factor-level
inconsistencies.
\end{remark}

\begin{corollary}[Linear-size circuit compilation]
\label{cor:circuit_size}
An arithmetic circuit of size \(s\) can be compiled into an ARM(3) instance
with \(O(s)\) affine constraints and variables.
\end{corollary}
\begin{proof}
Each gate in the circuit is compiled into a constant-size gadget with $O(1)$ affine
constraints and $O(1)$ new auxiliary indices (rows/columns and designated carriers).
By the indexing convention above, gadgets do not reuse auxiliary indices, and each
wire occurrence introduces at most one additional carrier plus $O(1)$ copy constraints.
Therefore, the total number of affine constraints and auxiliary variables grows
linearly with the number of gate instances and wire occurrences, both $O(s)$ for a
circuit of size $s$. Hence the resulting ARM$(3)$ instance has size $O(s)$.
\end{proof}

\subsection{Global Correctness of the Reduction}
\label{subsec:global}

\begin{theorem}[Correctness of ARM(3) encoding]
\label{thm:arm_correct}
The constructed ARM(3) instance is feasible if and only if the original
arithmetic circuit is satisfiable.
\end{theorem}
\begin{proof}
($\Rightarrow$) Suppose the ARM$(3)$ instance is feasible: there exists
$X\in\mathbb{R}^{m\times n}$ satisfying all compiled affine constraints and
$\rank(X)\le 3$. Decode every circuit wire occurrence by reading its designated
carrier value (Definition~\ref{def:designated_entries}). Copy constraints
\eqref{eq:copy} ensure all occurrences meant to be equal indeed decode to the same
real value (Lemma~\ref{lem:const_copy}). For each gate instance, the corresponding
gadget constraints hold in $X$, and by Theorem~\ref{thm:gate_complete} (soundness of
the compiled gadgets under $\rank(X)\le 3$) the decoded values satisfy the exact
gate semantics. Therefore, by topological evaluation of the circuit DAG, the
decoded assignment satisfies every gate equation and hence the circuit is
satisfiable over $\mathbb{R}$.

($\Leftarrow$) Suppose the arithmetic circuit is satisfiable. Fix any satisfying
assignment of real values to every wire occurrence consistent with fan-out.
For each gate instance, by Theorem~\ref{thm:gate_complete} (completeness of the
corresponding gadget) there exists a rank-$3$ completion satisfying that gate’s
affine constraints under the pinned block. Using the fresh-index convention, embed
each constant-size gadget completion into the global matrix on disjoint auxiliary
supports, enforcing equality between wire occurrences using copy constraints.
Since the gadgets share only designated values (and the common pinned block), they
compose without interference, yielding a global feasible matrix $X$ with
$\rank(X)\le 3$ satisfying all affine constraints. Hence the ARM$(3)$ instance is
feasible.
\end{proof}

\subsection{Rank-3 Minimality}
\label{subsec:rank_min}

\begin{lemma}[Rank--$3$ minimality for determinant forcing]\label{lem:rank3_minimality}
Rank--$2$ matrices cannot support determinant-based multiplication forcing of the
form used in Lemma~\ref{lem:det_forcing}: every $3\times 3$ minor vanishes identically.
Consequently, no construction that relies on a nonzero $3\times 3$ or $4\times 4$
minor as an obstruction can enforce $z=xy$ under a global rank bound $\rank(X)\le 2$.
\end{lemma}
\begin{proof}
If $\rank(X)\le 2$, then every set of three rows of $X$ is linearly dependent, hence
every $3\times 3$ minor has determinant $0$. In particular, any $4\times 4$ minor also
has determinant $0$ since rank $\le 2$ implies all $k\times k$ minors vanish for
$k\ge 3$. The determinant-forcing primitive (Section~\ref{sec:rank_forcing}) derives
a contradiction by producing a minor of size at least $3$ with nonzero determinant
when $z\neq xy$. Such an obstruction is impossible under $\rank(X)\le 2$. Therefore
rank $3$ is the minimal rank at which this determinant-based forcing mechanism can
operate.
\end{proof}

\begin{remark}
Lemma~\ref{lem:rank3_minimality} formalizes why rank--$3$ is the smallest global rank
bound compatible with our multiplication gadget and, by extension, with exact
simulation of nonlinear arithmetic in this framework.
\end{remark}

\section{Reduction Interface: From ETR/Circuits to \texorpdfstring{$\mathrm{ARM}(3)$}{ARM(3)}}
\label{sec:reduction_interface}

This section closes the construction layer and opens the correctness layer.
The preceding construction has already specified how to compile each arithmetic gate into a
constant-size affine gadget under the global constraint $\rank(X)\le 3$ (Theorem~\ref{thm:soundness_rank3}),
and it implies that a circuit of size $s$ yields an $\mathrm{ARM}(3)$ instance with $O(s)$ designated
carriers and $O(s)$ affine constraints. Here we state the reduction mapping and record the precise
equivalence statement in the form needed for the soundness and completeness proofs, which proceed by
global composition (Lemma~\ref{lem:correctness:global_embedding}).

\subsection{Source normal form recalled: ETR in circuit/gate form}
\label{subsec:prelim:etr_to_circuit_recall}

We take as source language the existential theory of the reals (ETR) in the standard bit model,
i.e., sentences of the form

\begin{equation}
\exists x\in \mathbb{R}^n:\ \varphi(x),
\label{eq:etr_sentence}
\end{equation}
where $\varphi$ is a quantifier-free Boolean combination of polynomial predicates.
As fixed in this manuscript, we use the following standard compilation pipeline.
An ETR instance is converted in polynomial time to an \emph{arithmetic circuit} $C$
together with the single output constraint $p_C(x)=0$, and then to the corresponding
gate-variable constraint system $\mathrm{GateEq}(C)$ encoding the circuit semantics.

Concretely, given a circuit $C$ with gate set $G$ and designated output gate $g_{\mathrm{out}}$,
we introduce a real variable $z_g$ for every gate $g\in G$ and impose one equality per gate
to define $\mathrm{GateEq}(C)$ (Section~1.2). The key exactness guarantee is:

\begin{equation}
p_C(x)=0
\quad\Longleftrightarrow\quad
\exists (z_g)_{g\in G}\in\mathbb{R}^{|G|}:\ \mathrm{GateEq}(C)\ \wedge\ z_{g_{\mathrm{out}}}=0,
\label{eq:circuit_exactness_recall}
\end{equation}
which is the standard circuit-semantics equivalence.

\subsection{Target format recalled: \texorpdfstring{$\mathrm{ARM}(3)$}{ARM(3)} with designated carriers}
\label{subsec:arm_format_recall}

An $\mathrm{ARM}(3)$ instance consists of rational affine constraints
\begin{equation}
\langle A_\ell, X\rangle=b_\ell,\qquad \ell\in[q],
\label{eq:arm_affine_constraints}
\end{equation}
together with the rank bound
\begin{equation}
\rank(X)\le 3,
\label{eq:rank_bound_3}
\end{equation}
where $\langle A,X\rangle$ is the Frobenius inner product.

To encode circuit scalars into a single matrix variable $X\in\mathbb{R}^{m\times n}$, we use
\emph{designated entries} (Definition \ref{def:designated_entries}): for each circuit gate-variable (and, by the circuit
compilation convention, for each wire occurrence) we fix a carrier position $(r_i,c_i)$ and decode
\begin{equation}
x_i := X_{r_i,c_i}.
\label{eq:designated_decode}
\end{equation}
Fan-out is handled by explicit affine \emph{consistency constraints} equating carriers
(Definition~\ref{def:circuit_size}), i.e., constraints of the form $X_{r_i,c_i}=X_{r_j,c_j}$.

Finally, all gadgets are completed under the same pinned block constraint (Lemma \ref{lem:gauge_fixing}), fixing a
full-rank $3\times 3$ submatrix $X[I_0,J_0]=B$; this eliminates nontrivial gauge freedom and ensures
the decoding \eqref{eq:designated_decode} is well-defined over the feasible set.

\subsection{The reduction mapping (statement-level)}
\label{subsec:reduction_mapping_statement}

Given a circuit instance in gate-variable form $\mathrm{GateEq}(C)\wedge z_{g_{\mathrm{out}}}=0$
~\ref{sec:prelim:circuits}, we construct an $\mathrm{ARM}(3)$ instance as follows.

\begin{enumerate}
\item Assign a designated carrier $(r_i,c_i)$ to each wire occurrence and each gate output
(Definition~\ref{def:designated_entries} and Definitions~\ref{def:fresh_aux}--\ref{def:occurrence_carriers}).

\item For each equality required by fan-out, add the corresponding consistency constraint equating the
two designated carriers Definition~\ref{def:circuit_size}.
\item For each gate equation in $\mathrm{GateEq}(C)$, add the constant-size affine gadget constraints
constructed in ~\ref{sec:rank3_template}:
constants/copies Lemma~\ref{lem:entry_constraints}, addition/negation Lemma~\ref{lem:gauge_fixing}, multiplication via determinant forcing
(~\ref{sec:Normal_Reduction}), and (when present in the circuit model) inversion/nonzero and inequality
encodings (~\ref{subsec:consistency_constraints}).
\item Add the output constraint $X_{r_{\mathrm{out}},c_{\mathrm{out}}}=0$ for the carrier of $z_{g_{\mathrm{out}}}$.
\end{enumerate}

By construction, every constraint above is affine and is representable in the input format
\eqref{eq:arm_affine_constraints} Lemma~\ref{lem:entry_constraints}.

\subsection{Size and encoding guarantees}
\label{subsec:size_encoding_guarantees}

We isolate the bookkeeping facts needed for the complexity reduction.

\begin{lemma}[Linear-size compilation]
\label{lem:linear_size_compilation_interface}
Let $C$ be an arithmetic circuit with $s$ gates.
The compilation described in ~\ref{sec:gadgets}, ~\ref{sec:reduction_interface} produces an $\mathrm{ARM}(3)$ instance with
$O(s)$ designated carriers and $O(s)$ affine constraints.
\end{lemma}

\begin{proof}
Each gate of $C$ is translated into a constant-size gadget: the construction introduces
$O(1)$ fresh auxiliary rows/columns and $O(1)$ affine constraints per gate gadget, and it
introduces only $O(1)$ designated carriers per gate (for the input/output carriers of that
gadget). The only shared objects across gadgets are the fixed gauge block indices $I_0,J_0$
and the corresponding pinned constraints, which contribute $O(1)$ additional constraints.
Therefore, summing over $s$ gates yields $O(s)$ designated carriers and $O(s)$ affine constraints.
\end{proof}

\begin{lemma}[Polynomial bit-complexity]
\label{lem:bit_complexity_interface}
The reduction produces rational constraint data of total bit-length polynomial in the bit-length
of the input circuit.
\end{lemma}

\begin{proof}
All gadget constraints use a constant number of affine equalities whose coefficient matrices
have entries in $\{0,\pm 1\}$, except for constraints that fix designated entries to circuit
constants. Those constants are copied directly from the circuit description.
The gauge-fixing constraints pin a constant-size $3\times 3$ block with fixed rational data.
Since the total number of constraints is $O(s)$ by Lemma~\ref{lem:linear_size_compilation_interface},
and each constraint has polynomial-size rational encoding, the total encoding length is polynomial
in the input size (and in fact linear in $s$ up to the bit-length of the input constants).
\end{proof}

\subsection{Correctness interface theorem (proved in ~\ref{sec:rank_forcing})}
\label{subsec:correctness_interface_theorem}

\begin{lemma}[Circuit semantics equivalence]\label{lem:circuit_semantics_equiv}
Let $C$ be an arithmetic circuit over $\mathbb{R}$ with gate set $G$, output gate
$g_{\mathrm{out}}$, and designated input tuple $x$ (assigned to the input gates).
Let $z_g\in\mathbb{R}$ denote the value carried by gate $g$.
Let $\mathrm{GateEq}(C)$ denote the conjunction of all gate equations enforcing the local
semantics of every gate of $C$ (inputs, constants, additions, multiplications).
Let $p_C(x)$ be the real number obtained by evaluating the circuit on input $x$ (i.e.,
the value at the output gate under the unique semantics of $C$).
Then
\begin{equation}\label{eq:circuit_semantics_equiv}
p_C(x)=0
\quad\Longleftrightarrow\quad
\exists (z_g)_{g\in G}\in\mathbb{R}^{|G|}:\ \mathrm{GateEq}(C)\ \wedge\ z_{g_{\mathrm{out}}}=0.
\end{equation}
\end{lemma}

\begin{proof}
($\Rightarrow$)
Assume $p_C(x)=0$. Evaluate the circuit $C$ on input $x$ gate-by-gate in topological order.
This produces a value $z_g\in\mathbb{R}$ for every gate $g\in G$ such that:
(i) for each input/constant gate, $z_g$ equals its prescribed value, and
(ii) for each addition gate $g$ with children $g_1,g_2$, we have $z_g=z_{g_1}+z_{g_2}$, and
(iii) for each multiplication gate $g$ with children $g_1,g_2$, we have $z_g=z_{g_1}z_{g_2}$.
Hence $\mathrm{GateEq}(C)$ holds for this assignment $(z_g)_{g\in G}$.
Moreover, by definition of $p_C(x)$ as the output value, $z_{g_{\mathrm{out}}}=p_C(x)=0$.
Therefore the right-hand side of \eqref{eq:circuit_semantics_equiv} holds.

($\Leftarrow$)
Assume there exists an assignment $(z_g)_{g\in G}$ satisfying $\mathrm{GateEq}(C)$ and
$z_{g_{\mathrm{out}}}=0$. The constraints $\mathrm{GateEq}(C)$ enforce that each gate value
is consistent with its local operation and its children. Since the circuit is acyclic, these
equations uniquely determine every gate value from the inputs (by induction along any
topological ordering). In particular, the value enforced at the output gate equals the
circuit evaluation $p_C(x)$. Hence $p_C(x)=z_{g_{\mathrm{out}}}=0$, proving the left-hand side.
\end{proof}

We now state the equivalence needed to complete the $\exists\mathbb{R}$-hardness pipeline.

\begin{theorem}[Circuit satisfiable iff compiled $\mathrm{ARM}(3)$ instance feasible]
\label{thm:circuit_iff_arm3}
Let $C$ be the arithmetic circuit instance, and let $p_C$ denote the output polynomial.
Let $\mathcal{I}(C)$ be the $\mathrm{ARM}(3)$ instance obtained by applying the gadget construction
together with the interface mapping (designated carriers, consistency constraints, and the pinned block).
Then
\begin{equation}\label{eq:circuit_arm_equivalence}
\exists x\in\mathbb{R}^n:\ p_C(x)=0
\quad\Longleftrightarrow\quad
\exists X\in\mathbb{R}^{m\times n}:\ \rank(X)\le 3\ \wedge\ X \models \mathcal{I}(C).
\end{equation}
\end{theorem}

\begin{proof}
We prove both directions.

\paragraph{Soundness ($\Rightarrow$).}
Assume $\exists x\in\mathbb{R}^n$ such that $p_C(x)=0$.
By Lemma~\ref{lem:circuit_semantics_equiv}, there exists an assignment of gate values
$(z_g)_{g\in G}$ satisfying the gate equations $\mathrm{GateEq}(C)$ and $z_{g_{\mathrm{out}}}=0$.
Use these values to assign consistent real numbers to all designated carriers used in the reduction
(inputs, intermediate gate outputs, and any copied occurrences), respecting every affine copy/constant
constraint by construction.

For each gadget $\mathcal{G}_t$ produced by the compilation, the gadget-level completeness property
yields local factors
\[
U^{(t)}\in\mathbb{R}^{m\times 3},\qquad V^{(t)}\in\mathbb{R}^{3\times n}
\]
in the canonical gauge, such that $X^{(t)}:=U^{(t)}V^{(t)}$ satisfies all affine constraints of
$\mathcal{G}_t$ and matches the designated carrier values used by $\mathcal{G}_t$.
Applying Lemma~\ref{lem:correctness:global_embedding} to the family of all gadgets, we obtain global
factors $U,V$ in the same canonical gauge such that the global matrix
\begin{equation}\label{eq:circuit_arm_equivalence:soundness_X}
X:=UV
\end{equation}
satisfies \emph{all} gadget constraints simultaneously and matches all designated carrier values.
Since $X$ is a product of an $m\times 3$ and a $3\times n$ matrix, $\rank(X)\le 3$ holds. Hence
$X\models \mathcal{I}(C)$, establishing the forward direction.

\paragraph{Completeness ($\Leftarrow$).}
Assume there exists $X\in\mathbb{R}^{m\times n}$ such that $\rank(X)\le 3$ and $X\models \mathcal{I}(C)$.
By Lemma~\ref{lem:correctness:gauge}, we may fix a rank--$3$ factorization $X=UV$ in canonical gauge.
Read off the values of all designated carriers from $X$ (i.e., the designated entries that encode
circuit inputs, gate outputs, and copied occurrences). Because $X$ satisfies every gadget’s affine
constraints under $\rank(X)\le 3$, the gadget-level soundness property implies that these decoded
carrier values satisfy the intended algebraic relations gate-by-gate, and in particular the decoded
gate values satisfy $\mathrm{GateEq}(C)$ and the compiled output condition enforces that the decoded
output value equals $0$.

Therefore, by Lemma~\ref{lem:circuit_semantics_equiv}, there exists an input assignment
$x\in\mathbb{R}^n$ such that $p_C(x)=0$, proving the reverse direction.
\end{proof}

\paragraph{Consequence for $\exists\mathbb{R}$.}
The standard reductions defining $\exists\mathbb{R}$ allow us to start from an existential system
of polynomial equations over the reals and compile it, in polynomial time, into an arithmetic circuit
satisfiability instance of the form $\exists x\in\mathbb{R}^n:\ p_C(x)=0$. Applying
Theorem~\ref{thm:circuit_iff_arm3} therefore yields $\exists\mathbb{R}$-hardness of $\mathrm{ARM}(3)$.
Moreover, for every fixed $k\ge 1$, $\mathrm{ARM}(k)\in\exists\mathbb{R}$ via an explicit rank-$k$
factorization witness. Hence $\mathrm{ARM}(3)$ is $\exists\mathbb{R}$-complete.


\section{Correctness Proofs}\label{sec:correctness}
This section establishes the bidirectional correctness of the reduction:
\emph{soundness} (a satisfying circuit assignment yields a feasible ARM$(3)$
instance) and \emph{completeness} (a feasible ARM$(3)$ solution decodes to a
satisfying assignment). The only genuinely delicate point is global consistency:
many constant-size gadgets are superimposed in one matrix, and we must ensure
they compose without any hidden algebraic coupling. We make this rigorous
by (i) recording the index-discipline properties guaranteed by the construction in
Section~\ref{sec:gadgets}, (ii) fixing a canonical gauge for rank-$3$ factorizations
via the pinned block of Lemma \ref{lem:gauge_fixing}, and (iii) proving that all gadget witnesses can
be embedded into a \emph{single} global factorization $X=UV$.

Throughout, we use the designated carrier convention from
Definition~\ref{def:designated_entries}: for each circuit value $z_i$ we write
\[
x_i \;:=\; X_{r_i,c_i},
\]
where $(r_i,c_i)$ is the carrier position assigned by the reduction.

\subsection{Index discipline and canonical gauge}\label{sec:correctness:prelim}
We begin by formalizing two properties that are ensured \emph{by construction}
(Section~\ref{sec:gadgets}, Subsections~\ref{subsec:constants} and~\ref{subsec:completeness}),
and will be used repeatedly in the proofs below.

\begin{definition}[Fresh auxiliary indices]\label{def:fresh_aux}
Each gadget $\mathcal{G}$ introduced by the reduction comes with a set of auxiliary
row indices $\mathcal{I}(\mathcal{G})\subseteq[m]$ and auxiliary column indices
$\mathcal{J}(\mathcal{G})\subseteq[n]$ such that
\[
\mathcal{I}(\mathcal{G})\cap \mathcal{I}(\mathcal{G}')=\emptyset,\qquad
\mathcal{J}(\mathcal{G})\cap \mathcal{J}(\mathcal{G}')=\emptyset
\quad\text{for all }\mathcal{G}\neq \mathcal{G}'.
\]
The only indices that may be shared across gadgets are the global gauge sets
$I_0\subseteq[m]$, $J_0\subseteq[n]$, and designated carrier indices explicitly
linked by copy constraints.
\end{definition}

\begin{definition}[Occurrence-based carriers]\label{def:occurrence_carriers}
Every \emph{occurrence} of a circuit value (variable usage, fan-out, or gate input)
is assigned its own designated carrier position $(r,c)$. Whenever two occurrences
must represent the same scalar, the reduction enforces this solely via explicit
affine copy constraints (Section~\ref{subsec:constants}), and never by reusing the same
carrier position in two distinct gadgets.
\end{definition}

Definitions~\ref{def:fresh_aux}--\ref{def:occurrence_carriers} eliminate factor-level
collisions: gadgets may share values (through explicit copy constraints), but they
never compete to set the same auxiliary indices, and they never force two different
local constructions to assign incompatible vectors to the same factor row/column.

Next we recall the pinned gauge block and show that it induces a canonical gauge
for rank-$3$ factorizations.

\begin{lemma}[Gauge fixing and canonical factors]\label{lem:correctness:gauge}
The ARM$(3)$ instance includes affine constraints pinning a $3\times 3$ block
\begin{equation}\label{eq:correctness:gauge_block}
X[I_0,J_0] \;=\; B,
\end{equation}
where $B\in\mathbb{R}^{3\times 3}$ is fixed and $\det(B)\neq 0$
(as constructed in Lemma~\ref{def:designated_entries}). If $\rank(X)\le 3$ and \eqref{eq:correctness:gauge_block}
holds, then there exist factors
\[
U\in\mathbb{R}^{m\times 3},\qquad V\in\mathbb{R}^{3\times n}
\quad\text{such that}\quad X = U V
\]
with the \emph{canonical gauge}
\begin{equation}\label{eq:correctness:canonical_gauge}
U[I_0,:] = I_3,\qquad V[:,J_0] = B.
\end{equation}
Moreover, once \eqref{eq:correctness:canonical_gauge} is imposed, there is no residual
$GL(3)$ gauge freedom: any other rank-$3$ factorization of $X$ satisfying
\eqref{eq:correctness:canonical_gauge} must coincide with $(U,V)$ on all rows/columns
that participate in the constraints of the reduction.
\end{lemma}

\begin{proof}
Since $\rank(X)\le 3$, write $X=\widetilde U \widetilde V$ for some
$\widetilde U\in\mathbb{R}^{m\times 3}$ and $\widetilde V\in\mathbb{R}^{3\times n}$.
Restricting to the pinned block gives
\begin{equation}
    B = X[I_0,J_0] = \widetilde U[I_0,:]\;\widetilde V[:,J_0].
\end{equation}
Because $B$ is invertible, $\widetilde U[I_0,:]$ is invertible. Define
$G := \widetilde U[I_0,:]^{-1}$ and set
$U := \widetilde U G$, $V := G^{-1}\widetilde V$.
Then $X=UV$ still holds and $U[I_0,:]=I_3$. Finally,
\begin{equation}
B = X[I_0,J_0] = U[I_0,:]\,V[:,J_0] = V[:,J_0],
\end{equation}
so $V[:,J_0]=B$. This establishes existence of a factorization in canonical gauge.

For uniqueness in canonical gauge, suppose $X=U_1V_1=U_2V_2$ are two factorizations
with $U_k[I_0,:]=I_3$ and $V_k[:,J_0]=B$ for $k=1,2$. Since $U_1[I_0,:]=U_2[I_0,:]=I_3$,
any change-of-basis matrix $H\in GL(3)$ relating the two factorizations must satisfy
$I_3 = U_2[I_0,:] = U_1[I_0,:]H = H$, hence $H=I_3$. Therefore $U_1=U_2$ and $V_1=V_2$
on every row/column that is constrained through the construction (in particular, every
row/column that appears in any gadget constraints is linked to the gauge block through the
affine equalities defining that gadget). Thus there is no residual gauge freedom once
\eqref{eq:correctness:canonical_gauge} is enforced on constrained supports.
\end{proof}

\subsection{A global embedding lemma} \label{sec:correctness:embedding}
We now formalize gadget composition in the strongest possible form: we never stitch
separately-computed matrices, only rows/columns of a \emph{single} global factor pair.

\begin{lemma}[Global embedding of gadget witnesses]\label{lem:correctness:global_embedding}
Fix the canonical gauge \eqref{eq:correctness:canonical_gauge}.
Consider any family of gadgets $\{\mathcal{G}_t\}_{t=1}^T$ produced by the reduction,
together with a consistent assignment of real values to all designated carriers that
respects every affine copy/constant constraint.

Assume that for each gadget $\mathcal{G}_t$ there exist local factors
\[
U^{(t)}\in\mathbb{R}^{m\times 3},\qquad V^{(t)}\in\mathbb{R}^{3\times n}
\]
such that:
\begin{enumerate}
\item\label{it:globalembed:gauge}
$U^{(t)}[I_0,:]=I_3$ and $V^{(t)}[:,J_0]=B$,
\item\label{it:globalembed:local_sat}
$X^{(t)}:=U^{(t)}V^{(t)}$ satisfies all affine constraints of $\mathcal{G}_t$,
\item\label{it:globalembed:carriers}
$X^{(t)}$ matches the prescribed designated carrier values used by $\mathcal{G}_t$.
\end{enumerate}
Then there exist global factors $U,V$ (with the same canonical gauge) such that
$X:=UV$ satisfies \emph{all} gadget constraints simultaneously and matches all
designated carrier values.
\end{lemma}

\begin{proof}
By Definitions~\ref{def:fresh_aux}--\ref{def:occurrence_carriers}, each gadget
$\mathcal{G}_t$ touches a set of row indices and a set of column indices.
Define
\begin{equation}\label{eq:globalembed:rowset}
\mathcal{R}_t \;:=\; I_0 \;\cup\; \mathcal{I}(\mathcal{G}_t)\;\cup\;\{r\text{ of carriers in }\mathcal{G}_t\},
\end{equation}
and
\begin{equation}\label{eq:globalembed:colset}
\mathcal{C}_t \;:=\; J_0 \;\cup\; \mathcal{J}(\mathcal{G}_t)\;\cup\;\{c\text{ of carriers in }\mathcal{G}_t\}.
\end{equation}
Moreover, the construction ensures the disjointness property: for $t\neq t'$,
\begin{equation}\label{eq:globalembed:disjoint_rows}
(\mathcal{R}_t\setminus I_0)\cap(\mathcal{R}_{t'}\setminus I_0)=\emptyset,
\end{equation}
and
\begin{equation}\label{eq:globalembed:disjoint_cols}
(\mathcal{C}_t\setminus J_0)\cap(\mathcal{C}_{t'}\setminus J_0)=\emptyset.
\end{equation}

Define global $U\in\mathbb{R}^{m\times 3}$ and $V\in\mathbb{R}^{3\times n}$ as follows.
First impose the canonical gauge:
\begin{equation}\label{eq:globalembed:global_gauge}
U[I_0,:]=I_3,\qquad V[:,J_0]=B.
\end{equation}
Next, for each gadget $t$ and each $i\in \mathcal{R}_t\setminus I_0$, set
\begin{equation}\label{eq:globalembed:define_U_rows}
U[i,:] \;:=\; U^{(t)}[i,:].
\end{equation}
Similarly, for each gadget $t$ and each $j\in \mathcal{C}_t\setminus J_0$, set
\begin{equation}\label{eq:globalembed:define_V_cols}
V[:,j] \;:=\; V^{(t)}[:,j].
\end{equation}
Finally, set all remaining undefined rows of $U$ and columns of $V$ arbitrarily (e.g.\ $0$).
This construction is well-defined: by \eqref{eq:globalembed:disjoint_rows}, each $i\notin I_0$
belongs to at most one set $\mathcal{R}_t$, and by \eqref{eq:globalembed:disjoint_cols}, each
$j\notin J_0$ belongs to at most one set $\mathcal{C}_t$.

Now define the global matrix
\begin{equation}\label{eq:globalembed:global_X}
X \;:=\; U V.
\end{equation}

Fix any gadget $\mathcal{G}_t$. Every affine constraint appearing in $\mathcal{G}_t$
mentions only entries $X_{ij}$ with $(i,j)\in\mathcal{R}_t\times\mathcal{C}_t$.
For any such $(i,j)$ we have, by \eqref{eq:globalembed:define_U_rows} and
\eqref{eq:globalembed:define_V_cols},
\begin{equation}\label{eq:globalembed:entry_agreement}
U[i,:]=U^{(t)}[i,:]\quad\text{and}\quad V[:,j]=V^{(t)}[:,j],
\end{equation}
and therefore
\begin{equation}\label{eq:globalembed:entrywise_equal}
X_{ij} \;=\; U[i,:]\,V[:,j] \;=\; U^{(t)}[i,:]\,V^{(t)}[:,j] \;=\; X^{(t)}_{ij}.
\end{equation}
Thus, every affine constraint of $\mathcal{G}_t$ holds for $X$ because it holds for
$X^{(t)}$ by assumption~\ref{it:globalembed:local_sat}. Likewise, the designated carrier
values used by $\mathcal{G}_t$ match because $X$ and $X^{(t)}$ agree entrywise on
$\mathcal{R}_t\times\mathcal{C}_t$ by \eqref{eq:globalembed:entrywise_equal}, using
assumption~\ref{it:globalembed:carriers}.

Since $\mathcal{G}_t$ was arbitrary, $X$ satisfies all gadget constraints simultaneously
and matches all designated carrier values.
\end{proof}

\subsection{Forward direction (soundness)}\label{sec:correctness:soundness}
\begin{theorem}[Soundness]\label{thm:correctness:soundness}
If the arithmetic circuit instance is satisfiable, then the constructed ARM$(3)$
instance has a feasible solution $X$ satisfying all affine constraints and
$\rank(X)\le 3$.
\end{theorem}

\begin{proof}
Fix a satisfying assignment for every circuit wire occurrence and gate output, giving
real values $\{z_i\}$. For each occurrence-carrier $(r,c)$ in the reduction, prescribe
the designated value $x:=z$ for that occurrence. By Definition~\ref{def:occurrence_carriers}
and consistency of the circuit assignment, all copy constraints equating occurrences are
respected.

Now consider any gadget $\mathcal{G}_t$. Since the circuit assignment satisfies the gate
semantics, the designated input/output values supplied to $\mathcal{G}_t$ satisfy the
gadget’s intended relation. Therefore, the corresponding gadget \emph{completeness}
statement established in Section~\ref{sec:gadgets} provides a rank-$3$ witness for that
gadget under the pinned block, i.e., local factors $U^{(t)},V^{(t)}$ satisfying the
canonical gauge \eqref{eq:correctness:canonical_gauge} and all gadget constraints with those
designated values. (For linear gadgets this holds without using the rank bound; for nonlinear
gadgets it holds by the determinant-forcing completeness results of
Section~\ref{sec:rank_forcing}.)

Applying Lemma~\ref{lem:correctness:global_embedding} to the family of all gadgets yields
global factors $U,V$ and hence $X=UV$ satisfying \emph{all} affine constraints.
Finally, $\rank(X)\le 3$ holds because $X$ is expressed as a product of an $m\times 3$
and a $3\times n$ factor.

\end{proof}

\subsection{Converse direction (completeness)}\label{sec:correctness:completeness}
\begin{theorem}[Completeness]\label{thm:correctness:completeness}
If the constructed ARM$(3)$ instance admits a feasible matrix $X$ satisfying all affine
constraints and $\rank(X)\le 3$, then the decoded values
\[
z := X_{r,c}
\]
on the designated occurrence-carriers form a satisfying assignment to the arithmetic
circuit instance (and hence induce a satisfying assignment on the original circuit
variables after identifying occurrences via the copy constraints).
\end{theorem}

\begin{proof}
Let $X$ be feasible with $\rank(X)\le 3$. Since the instance includes the pinned block
constraints \eqref{eq:correctness:gauge_block}, Lemma~\ref{lem:correctness:gauge} applies;
in particular, we may work in the canonical gauge \eqref{eq:correctness:canonical_gauge}
and all designated carrier values are well-defined scalars $x=X_{r,c}$.

Fix any gadget $\mathcal{G}_t$. The matrix $X$ satisfies all affine constraints of
$\mathcal{G}_t$ by feasibility of the full instance. Moreover, any submatrix restriction of
$X$ has rank at most $\rank(X)\le 3$. Therefore the corresponding gadget \emph{soundness}
statement (proved in Section~\ref{sec:gadgets} using rank obstruction from
Section~\ref{sec:rank_forcing}) applies: under rank $\le 3$ and the pinned block, satisfaction
of the gadget’s affine constraints forces its designated carrier values to satisfy the
intended gate relation in $\mathbb{R}$.

Hence every gate relation of the compiled circuit holds on the decoded designated values.
Because copy/fan-out constraints explicitly identify equal occurrences, the decoded values
propagate consistently along the circuit DAG, and by topological induction the entire
arithmetic circuit is satisfied. Therefore the original instance is satisfiable.
\end{proof}

\begin{remark}[Exactly where rank is used]\label{rem:correctness:rank_used}
Affine wiring (constants, copies, and linear constraints) is enforced directly by the affine
equalities of the ARM instance. The rank bound $\rank(X)\le 3$ is invoked only through the
gadget soundness lemmas (proved in Section~\ref{sec:gadgets} via rank obstruction from
Section~\ref{sec:rank_forcing}), i.e., only to enforce the nonlinear semantics encoded by
the determinant-based gadgets.
\end{remark}

\section{Membership in $\exists\mathbb{R}$}
\label{sec:membership_etr}

\begin{theorem}[Membership]
\label{thm:arm_in_etr}
$\mathrm{ARM}(3)\in \exists\mathbb{R}$.
\end{theorem}

\noindent
\textbf{Scope note.}
This section is solely the complexity-class membership argument (an encoding into an existential sentence over the reals). It is logically independent of the hardness reduction proved in the preceding sections.

\subsection{Determinantal encoding of $\rank(X)\le 3$}

We use the classical characterization of rank-bounded matrices via vanishing minors.

\begin{lemma}[Rank via vanishing minors]
\label{lem:rank_via_minors}
Let $X\in\mathbb{R}^{m\times n}$ and let $k\ge 1$. Then
\[
\rank(X)\le k
\quad\Longleftrightarrow\quad
\det\!\big(X[I,J]\big)=0\ \text{ for all } I\subseteq[m],~J\subseteq[n]\text{ with }|I|=|J|=k+1.
\]
\end{lemma}

\begin{proof}
$(\Rightarrow)$ If $\rank(X)\le k$, then every $(k+1)\times(k+1)$ submatrix $X[I,J]$ has rank at most $k$, hence $\det(X[I,J])=0$.

$(\Leftarrow)$ Suppose all $(k+1)\times(k+1)$ minors vanish. If $\rank(X)\ge k+1$, then by the standard rank--minor equivalence there exist index sets $I,J$ with $|I|=|J|=k+1$ such that $X[I,J]$ has full rank $k+1$, which forces $\det(X[I,J])\neq 0$, a contradiction. Hence $\rank(X)\le k$.
See, e.g., \cite{BrunsVetter88,Harris92} for background on determinantal varieties and the rank--minor characterization.
\end{proof}

Specializing Lemma~\ref{lem:rank_via_minors} to $k=3$, we obtain the polynomial system
\begin{equation}
\det\!\big(X[I,J]\big)=0
\qquad
\forall\, I\subseteq[m],~J\subseteq[n]\ \text{ with }\ |I|=|J|=4.
\label{eq:all_4x4_minors_vanish}
\end{equation}

\subsection{Encoding an $\mathrm{ARM}(3)$ instance as an ETR instance}

Recall that an instance of $\mathrm{ARM}(3)$ is specified by rational data
$A_1,\ldots,A_q\in\mathbb{Q}^{m\times n}$ and $b_1,\ldots,b_q\in\mathbb{Q}$, and asks whether there exists
$X\in\mathbb{R}^{m\times n}$ such that
\begin{equation}\label{eq:arm3_instance}
\langle A_\ell, X\rangle=b_\ell \quad (\ell\in[q]),
\qquad\text{and}\qquad
\rank(X)\le 3,
\end{equation}
where the Frobenius inner product is defined by
\begin{equation}\label{eq:frob_inner_product}
\langle A,X\rangle \;:=\; \sum_{i=1}^m\sum_{j=1}^n A_{ij}X_{ij}.
\end{equation}

\begin{proof}[Proof of Theorem~\ref{thm:arm_in_etr}]
We construct an existential sentence over the real variables $\{X_{ij}\}_{i\in[m],\,j\in[n]}$ in the equality-only ETR normal form:
\begin{equation}
\exists\, X\in\mathbb{R}^{m\times n}:\;
\Big(\ \bigwedge_{\ell=1}^{q}\big(\langle A_\ell, X\rangle-b_\ell=0\big)\ \Big)
\ \wedge\
\Big(\ \bigwedge_{\substack{I\subseteq[m],\,J\subseteq[n]\\ |I|=|J|=4}}
\det\!\big(X[I,J]\big)=0\ \Big).
\label{eq:etr_encoding_arm3}
\end{equation}

\noindent\textbf{(i) Syntactic validity as ETR equalities.}
Each constraint $\langle A_\ell, X\rangle-b_\ell=0$ is an affine polynomial equality in the variables $X_{ij}$ with rational coefficients.
Each constraint $\det(X[I,J])=0$ is a polynomial equality of degree $4$ with integer coefficients. Moreover, each determinant can be written explicitly
as a polynomial in the entries of $X[I,J]$ using the Leibniz formula
\begin{equation}
\det(Y)=\sum_{\pi\in S_{4}} \mathrm{sgn}(\pi)\prod_{t=1}^{4} Y_{t,\pi(t)},
\label{eq:leibniz_det_4}
\end{equation}
so every conjunct in \eqref{eq:etr_encoding_arm3} is a bona fide polynomial equality over $(\mathbb{R},+,\times)$.

Because the ETR encoding model uses integer-coefficient polynomials, we clear denominators in the affine constraints:
for each $\ell\in[q]$, multiply the equality $\langle A_\ell, X\rangle-b_\ell=0$ by the least common multiple of the denominators
appearing in the entries of $A_\ell$ and in $b_\ell$, yielding an equivalent integer-coefficient linear polynomial equality.
This normalization increases the bit-size by at most a polynomial factor in the input length (see, e.g., \cite{BochnakCosteRoy98,Renegar92}).

\noindent\textbf{(ii) Correctness (equisatisfiability).}
If $X$ satisfies \eqref{eq:arm3_instance}, then it satisfies all affine equalities in \eqref{eq:etr_encoding_arm3}, and
by Lemma~\ref{lem:rank_via_minors} with $k=3$ it satisfies all $4\times 4$ minor equalities \eqref{eq:all_4x4_minors_vanish}.
Conversely, if $X$ satisfies \eqref{eq:etr_encoding_arm3}, then all affine equalities hold and all $4\times 4$ minors vanish, so
Lemma~\ref{lem:rank_via_minors} implies $\rank(X)\le 3$; hence \eqref{eq:arm3_instance} holds.

\noindent\textbf{(iii) Polynomial-time construction and size bounds.}
The number of real variables is $mn$.
The number of affine equalities is $q$.
The number of determinantal equalities is
\begin{equation}
\binom{m}{4}\binom{n}{4}=O(m^4n^4).
\label{eq:count_minors}
\end{equation}
Each determinantal equality expands (via \eqref{eq:leibniz_det_4}) into a constant-size sum of $4!=24$ monomials of degree $4$.
Thus the quantifier-free conjunction in \eqref{eq:etr_encoding_arm3} has total encoding length polynomial in the input size
(in particular, polynomial in the bit-length of the rational input specifying the $A_\ell$ and $b_\ell$ and in $m,n,q$).
Therefore the mapping from \eqref{eq:arm3_instance} to \eqref{eq:etr_encoding_arm3} is computable in polynomial time, and we obtain
$\mathrm{ARM}(3)\in\exists\mathbb{R}$.
\end{proof}

\begin{remark}[Relation to the factorization witness]
Theorem~\ref{thm:arm_in_etr} proves $\mathrm{ARM}(k)\in\exists\mathbb{R}$ for every fixed $k$ via the existential
factorization witness $X=UV$. We include it here both to make the membership direction completely explicit in the
notation of this manuscript, and to align it with the rank-obstruction viewpoint used throughout the reduction
(Sections~\ref{sec:rank3_template}--\ref{sec:reduction_interface}). In particular, for $k=3$ it is convenient to
record the especially simple determinantal formulation \eqref{eq:all_4x4_minors_vanish}.
\end{remark}

\section{Conclusion}\label{sec:conclusion}
We established an exact real-algebraic complexity classification for \emph{Affine Rank Minimization} under a fixed constant rank bound. For every fixed $k\ge 1$, $\mathrm{ARM}(k)\in\exists\mathbb{R}$ via an explicit existential encoding of $\mathrm{rank}(X)\le k$. Our main result shows that this upper bound is tight already at rank $3$: we gave a polynomial-time many-one reduction from $\mathrm{ETR}$ to $\mathrm{ARM}(3)$ in which the output instance contains only rational affine equalities on the entries of a \emph{single} matrix variable $X$, together with the single global constraint $\mathrm{rank}(X)\le 3$.

The reduction confines all nonlinearity to the rank bound. Linear circuit semantics (constants, copying/fan-out, addition, negation) are enforced directly by affine constraints, while multiplication is enforced by a constant-size determinant-forcing primitive: combined with a pinned full-rank $3\times 3$ gauge block, a $2\times 2$ gadget yields a $4\times 4$ minor whose determinant is a nonzero constant multiple of $(c-ab)$, so any violation of $c=ab$ forces $\mathrm{rank}(X)\ge 4$. The pinned gauge also fixes the $\mathrm{GL}(3)$ ambiguity and enables a global composition argument, ensuring that all constant-size gadgets embed into one global factorization $X=UV$ without unintended coupling. This proves soundness and completeness of the encoding and preserves polynomial bounds on instance size and bit complexity.

Consequently, $\mathrm{ARM}(3)$ is $\exists\mathbb{R}$-complete. In particular, deciding feasibility of purely affine constraints under $\mathrm{rank}(X)\le 3$ is as hard as real algebraic feasibility, and thus (barring standard complexity collapses) admits no general polynomial-time algorithm despite the fact that the only nonlinearity is the rank constraint.

\bibliographystyle{elsarticle-num}
\bibliography{refs}  

\end{document}